\pdfoutput=1

\documentclass[12pt]{article}
\usepackage{amsmath,epsf,amssymb,latexsym,cite,eucal,amsthm,enumerate}
\usepackage[matrix,arrow,frame,import,curve,color]{xy}
\usepackage{tikz}
\usetikzlibrary{arrows}
\usetikzlibrary{decorations.pathmorphing}
\usetikzlibrary{decorations.markings}

\usepackage{graphicx}

\newtheorem{theorem}{Theorem}
\newtheorem{prop}[theorem]{Proposition}

\newtheorem{question}[theorem]{Question}

\theoremstyle{definition}
\newtheorem{definition}[theorem]{Definition}
\newtheorem{example}[theorem]{Example}

\newif\iffigs\figstrue

\DeclareFontFamily{U}{rsf}{}
\DeclareFontShape{U}{rsf}{m}{n}{
  <5> <6> rsfs5 <7> <8> <9> rsfs7 <10-> rsfs10}{}
\DeclareMathAlphabet\Scr{U}{rsf}{m}{n}

\def\a{{\mathsf a}}
\def\A{{\mathbf A}}
\def\B{{\mathbf B}}
\def\O{\Scr{O}}
\def\C{{\mathbb C}}
\def\P{{\mathbb P}}

\def\Q{{\mathbb Q}}
\def\R{{\mathbb R}}
\def\Z{{\mathbb Z}}
\def\H{{\mathbb H}}

\def\Hom{\operatorname{Hom}}

\def\Ext{\operatorname{Ext}}

\def\Aut{\operatorname{Aut}}

\def\coker{\operatorname{coker}}

\def\Spec{\operatorname{Spec}}

\def\SProj{\operatorname{\P roj}}

\def\Conv{\operatorname{Conv}}

\def\Pic{\operatorname{Pic}}
\def\Jac{\operatorname{Jac}}

\def\SU{\operatorname{SU}}

\def\rank{\operatorname{rank}}

\def\Cone{\operatorname{Cone}}

\def\id{{\mathbf{1}}}

\def\frak{\mathfrak}

\newcommand{\Moverbar}[1]{\mkern
  7mu\overline{\mkern-7mu#1\mkern+2.0mu}\mkern-2mu}

\def\CY{Calabi--Yau}
\def\LG{Landau--Ginzburg}

\def\GLSM{gauged linear $\sigma$-model}

\def\cM{{\Scr M}}
\def\cMcp{\Moverbar{\cM}}

\def\cA{{\Scr A}}
\def\cB{{\Scr B}}

\def\cI{{\Scr I}}

\def\DC{\mathbf{D}^b}

\def\cz{Z}

\def\RHom{\mathbf{R}\Hom}

\def\eqn#1#2{\begin{equation}#2
  \ifx{#1}{}\else\label{#1}\fi\end{equation}}

\def\fluffy{$\Q$-massless}

\begin{document}

\begin{titlepage}
\begin{flushright}
May 2013
\end{flushright}
\vspace{.5cm}
\begin{center}
\baselineskip=16pt
{\fontfamily{ptm}\selectfont\bfseries\huge
Categories of Massless D-Branes\\and del Pezzo Surfaces\\[20mm]}
{\bf\large  Nicolas Addington and Paul S.~Aspinwall
 } \\[7mm]

{\small

Department of Mathematics, 
  Box 90320 \\ Duke University, 
 Durham, NC 27708-0320 \\ \vspace{6pt}

 }

\end{center}

\begin{center}
{\bf Abstract}
\end{center}
In analogy with the physical concept of a massless D-brane, we define
a notion of ``\fluffy ness'' for objects in the derived category. This is
defined in terms of monodromy around singularities in the stringy K\"ahler
moduli space and is relatively easy to study using ``spherical functors''.
We consider several examples in which del Pezzo surfaces and
other rational surfaces in \CY\ threefolds are contracted. For precisely
the del Pezzo surfaces that can be written as hypersurfaces in weighted
$\P^3$, the category of \fluffy\ objects is a ``fractional \CY'' category
of graded matrix factorizations.


\end{titlepage}

\vfil\break

\def\Cl{{L^\vee}}

\section{Introduction}    \label{s:intro}

The D-branes in the topological B-model on a \CY\ threefold $X$ are
expected to be described by the derived category $\DC(X)$
\cite{Doug:DC,AL:DC,me:TASI-D}. For a physical D-brane associated to a
boundary condition in a conformal field theory one further imposes a
stability condition \cite{DFR:stab,Doug:DC} which varies as one
moves in the stringy moduli space of complexified K\"ahler forms, denoted $\cM$. Alternatively one may consider a space of stability conditions directly \cite{Brg:stab} and then view $\cM$ a subspace of this. 

If $\cMcp$ is a compactification of $\cM$, we can think of the
points in $\cMcp-\cM$ as ``bad conformal field theories''. There are
two reasons they can be bad:
\begin{enumerate}
\item They are ``at infinity'' with respect to the Zamolodchikov
  metric -- that is, they are ``large radius limits'' in some sense.
\item At least one D-brane has become massless \cite{Str:con} in a sense we now explain.
\end{enumerate}

A choice of complexified K\"ahler form gives a map $\cz:K(X)\to\C$ given by the central charge, which is key in determining stability conditions. The mass of a D-brane is proportional to $|\cz|$. Thus a massless D-brane is a \emph{stable} object for which $\cz\to0$ as we approach a point in $\cMcp-\cM$. Determining the set of stable D-branes is a notoriously difficult problem, so we propose an alternative approach to finding massless D-branes using monodromy.

Going around loops in $\cM$ induces automorphisms of $\DC(X)$ \cite{AD:Dstab}. If $y$ is a local coordinate on $\cM$ and $T \in \Aut(\DC(X))$ is the monodromy around a bad theory $y=0$, we define an object $\a \in \DC(X)$ to be ``\fluffy'' if
\[ T^q\a = \a[p] \]
for some for some positive integers $p$ and $q$.\footnote{We will also require that this not be true of \emph{every} object in $\DC(X)$, to avoid the degenerate case of a \LG\ orbifold theory.}  As we discuss later, this suggests that the central charge of $\a$ goes to zero as $y^{p/2q}$. In addition, \fluffy ness seems to capture some of the information about stability as $y\to0$.

The phase picture of the moduli space coming from the \GLSM\ \cite{W:phase} naturally categorizes monodromy as arising around ``phase limits'' or around the ``discriminant'' which asymptotically lives in walls dividing phases. Horja's EZ-spherical twists \cite{Horj:EZ} give a nice picture of the monodromy around the discriminant, but in this paper we will focus on the phase limits. In the simplest examples, the only massless D-branes to be seen occur at the discriminant, but we will see that in more general examples the phase limits can be far more interesting.

The main case we consider is that of a contracting del Pezzo surface $E \subset X$. Here the contracted threefold itself forms a phase limit although, as we discuss, the phase limit is better thought of as an ``exoflop''. The monodromy of the periods has already been carefully analyzed in this case \cite{Lerche:1996ni}. That is, the action of the monodromy on K-theory was computed using the GKZ differential equations. Here we study the monodromy action on the derived category itself using ``spherical functors'' \cite{MR2258045,Anno:sph}, which provide a very simple way of working with monodromy --- in many ways easier than solving the GKZ system.

Let $dP_n$ be $\P^2$ blown up at $n$ points in general position. This is a del Pezzo surface for $n\leq8$.\footnote{We will abuse notation somewhat and sometimes have $n>8$.} We summarize the properties we find when contracting a surface $E \subset X$:
\begin{enumerate}
\item If $E$ is $\P^2$, $\P^1\times\P^1$, $dP_6$, $dP_7$ or $dP_8$, then the category of \fluffy\ objects as $E$ collapses to a point is given by the subcategory $\A \subset \DC(E)$ appearing in the semiorthogonal decomposition
\begin{equation}
  \DC(E) = \langle \A,\O_E,\O_E(1),\ldots,\O_E(m-1)\rangle,
\label{eq:orth1}
\end{equation}
where $m$ is the index of $E$.  For $\P^2$ we have $m=3$ and $\A = 0$.  For $\P^1 \times \P^1$ we have $m=2$ and $\A = \DC(\text{2 points})$.  For the others we have $m=1$ and $\A$ is a category of graded matrix factorizations \cite{Orlov:mfc}, which is ``fractional \CY.''
\item If $E = dP_5$ then the category of \fluffy\ objects as $E$ collapses is a subcategory of $\A$ whose image in $K(\A)$ is codimension 1.  The ``other'' objects in $\A$ have monodromy consistent with their central charges going to zero as $\sqrt y \log y$, so all of $\A$ appears to go massless.
\item If $E\to\P^1$ is a conic fibration, with generic fibre $\P^1$ and $n$ degenerate fibres (of the form $\P^1\vee\P^1$), then as $E \subset X$ is collapses to $\P^1$ we find $2n$ \fluffy\ objects.
\item If $E = dP_n$ and $E \to \P^2$ is the blow-down of $n$ $(-1)$-curves, then as $E \subset X$ collapses to $\P^2$ we find $n$ \fluffy\ objects.  This is just the usual flop picture in which the $n$ curves collapse to $n$ points, but recast into our framework of spherical functors and semi-orthogonal decompositions.
\end{enumerate}

All our examples are complete intersections in toric varieties, so the geometry of $\cM$ is well-understood.  In section \ref{s:mon} we review enough of the toric machinery to be able to write our monodromy as $T_\infty = T_\Delta T_0$, where $T_\Delta$ is a Seidel-Thomas spherical twist or Horja EZ-spherical twist and $T_0$ is tensoring by a line bundle.  In section \ref{s:sphF} we review spherical functors and their interaction with semi-orthogonal decompositions, which will allow us work with this monodromy easily.  In section \ref{s:fluffy} we discuss \fluffy ness in detail, and review certain fractional \CY\ categories which will be our main source of \fluffy\ objects.  In section \ref{s:app} we turn to the examples mentioned above, which are the heart of the paper.  In section \ref{s:disc} we conclude with a discussion of unresolved questions.

\section{Monodromy} \label{s:mon}

\subsection{Mirror Pairs}  \label{ss:mir}

If a \CY\ threefold is given as a complete intersection in a normal
toric variety then the secondary fan may be used to construct a model
of $\cM$ \cite{Bat:m}. This construction is standard (see, for
example, \cite{AG:gmi}) but we need to review it to fix notation.

Let $N$ be a lattice of rank $d$ and $M$ the dual
lattice. Let $\cA \subset N$ be set of $n$ points such that
\begin{enumerate}
\item $\cA$ spans $N$.
\item $\cA = N \cap \Conv(\cA)$, that is, $\cA$ contains all the lattice points of its convex hull.
\item There is a $\mu \in M$ such that $\langle \mu,\alpha \rangle = 1$ for all $\alpha \in \cA$; in particular $\cA$ lies in an affine hyperplane.
\end{enumerate}
Let $C_\cA \subset N_\R$ be the cone over $\Conv(\cA)$ and let $C_{\cB}$ be the dual cone:
\begin{equation*}
C_{\cB} = \{s\in M_\R: \langle s,t\rangle\geq0 \text{ for all } t\in C_{\cA}\}.
\end{equation*}
Then we require a finite set of points $\cB \subset M$ such that $C_\cB$ is the cone over $\Conv(\cB)$ and
\begin{enumerate}
\item $\cB$ spans $M$.
\item $\cB = M \cap \Conv(\cB)$.
\item There is a $\nu \in N$ such that $\langle \beta,\nu \rangle = 1$ for all $\beta \in \cB$.
\end{enumerate}
This is a generalization of $\cA$ and $\cB$ being a ``reflexive pair'' in that we do not require $\langle \mu, \nu \rangle = 1$.  Among other things this allows us to work with complete intersections in toric varieties rather than just hypersurfaces.

The points of $\cA$ determine a surjective map $A:\Z^{\oplus n}\to N$, which we think of as a $d \times n$ integer matrix.  Form an exact sequence
\begin{equation}
\xymatrix@1@M=2mm{
  0\ar[r]&L\ar[r]&\Z^{\oplus n} \ar[r]^A&N \ar[r]& 0,
} \label{eq:LZN}
\end{equation}
where $L$ is the ``lattice of relations'' of rank $r=n-d$. Dual to this we have
\begin{equation}
\xymatrix@1@M=2mm{
  0\ar[r]&M\ar[r]&\Z^{\oplus n}\ar[r]^Q&\Cl\ar[r]& 0,
} \label{eq:M-D}
\end{equation}
where $Q$ is the $r\times n$ matrix of ``charges'' of the points in
$\cA$.

Let
\begin{equation*}
  S = \C[x_1,\ldots,x_n].
\end{equation*}
The matrix $Q$ gives an $r$-fold multi-grading to this ring: that is, we have a $(\C^*)^r$ torus action
\begin{equation*}
  x_\alpha \mapsto \lambda_1^{Q_{1\alpha}}\lambda_2^{Q_{2\alpha}}\ldots
              \lambda_r^{Q_{r\alpha}} x_\alpha,
\end{equation*}
where $\lambda_j\in\C^*$. Let $S_0$ be the $(\C^*)^r$-invariant
subalgebra of $S$. The algebra $S$ then decomposes into a sum of
$S_0$-modules labeled by their $r$-fold grading:
\begin{equation*}
  S = \bigoplus_{\mathbf{u}\in \Cl} S_{\mathbf{u}},
\end{equation*}
where $\Cl\cong\Z^{\oplus r}$ from (\ref{eq:M-D}).
As usual we denote a shift in grading by parentheses: $S(\mathbf{u})_{\mathbf{v}}=S_{\mathbf{u}+\mathbf{v}}$.

Consider a simplicial decomposition $\Sigma$ of the pointset $\cA$ which is {\em regular} in the sense of \cite{OP:convex}.  It may or may not include points in the interior of the convex hull of $\cA$. We refer
to a choice of simplicial decomposition as a ``phase''. 

To each phase we associate the ``Cox ideal'', defined in \cite{Cox:} as follows:
\begin{definition}
  The {\bf Cox ideal} $B_\Sigma$ is the ideal of $S$ generated by the monomials
\[ \prod \{ x_\alpha : \alpha \in \cA \setminus \text{vertices of  } \sigma \} \]
as $\sigma$ ranges over the maximal faces of $\Sigma$.
\end{definition}

We define the stack
\begin{equation*}
\begin{split}
  Z_\Sigma &= \SProj_\Sigma S\\
          &= [(\Spec S - V(B_\Sigma)) / (\C^*)^r ],
\end{split}
\end{equation*}
where $V(B_\Sigma)$ is the set of homogeneous prime ideals containing $B_\Sigma$. This is almost the same as the usual toric
variety associated to $\Sigma$. We use the stack terminology to
correctly deal with sheaves on orbifold singularities. This toric
version of a stack is exactly the same construction as that for a
weighted projective space explained in \cite{ALO:} except that
$B_\Sigma$ plays the r\^ole of the irrelevant ideal and $r$ can be
greater than 1. See also \cite{BCS:toric}.

\def\gr{\operatorname{gr}}
\def\tors{\operatorname{tors}}
\begin{prop}  \label{prop:DZ}
The category of coherent sheaves on $Z_\Sigma$ is given by the quotient
category $\gr(S)/\tors_\Sigma(S)$, where $\gr(S)$ is the category of
finitely-generated multigraded $S$-modules and $\tors_\Sigma(S)$ is
the subcategory of such modules annihilated by some power of $B_\Sigma$. \end{prop}
\begin{proof}
This follows exactly copying proposition 2.3 of \cite{ALO:}.
\end{proof}

Next we define the {\em superpotential\/}
\begin{equation*}
  W_X = \sum_{\beta\in\cB} a_\beta\, x_1^{\langle\alpha_1,\beta\rangle}
   x_2^{\langle\alpha_2,\beta\rangle}\cdots
    x_n^{\langle\alpha_n,\beta\rangle},
\end{equation*}
where $\cA=\{\alpha_1,\ldots,\alpha_n\}$ and $a_\beta\in\C$. Note that
this is homogeneous of multi-degree 0.

Define the Jacobian ideal
\begin{equation*}
\Jac(W_X) = \left(\frac{\partial W}{\partial
        x_1},\ldots,\frac{\partial W}{\partial
        x_n}\right),
\end{equation*}
and the corresponding stack
\begin{equation*}
  X_{\Sigma} = \SProj_\Sigma \frac{S}{\Jac(W_X)}.
\end{equation*}
Note that $X_\Sigma$ depends on a choice of the coefficients $a_\beta$ in the
superpotential. As a set, $X_\Sigma$ consists of the critical points of $W_X$.

We may now exchange the r\^oles of $\cA$ and $\cB$ to obtain a 
mirror superpotential
\begin{equation}
  W_Y = \sum_{\alpha\in\cA} b_\alpha\, y_1^{\langle\alpha,\beta_1\rangle}
   y_2^{\langle\alpha,\beta_2\rangle}\cdots
    y_m^{\langle\alpha,\beta_m\rangle},  \label{eq:Wb}
\end{equation}
for $\cB=\{\beta_1,\ldots,\beta_m\}$ and $b_\alpha\in\C$. Fixing a
regular simplicial decomposition $\Upsilon$ of $\cB$ gives a corresponding $Y_\Upsilon$ associated to the Jacobian of $W_Y$.

The statement of mirror symmetry for the gauged linear sigma model
\cite{W:phase} is that $X_\Sigma$ and $Y_\Upsilon$ form a mirror
pair. For the purposes of this paper we can use this as the {\em
  definition\/} of a mirror pair.

The standard example is given by the quintic threefold. Here the
matrix $A^t$, whose rows give the pointset $\cA$, is
\begin{equation*}
  A^t=\begin{pmatrix}1&0&0&0&0\\1&1&0&0&0\\1&0&1&0&0\\
    1&0&0&1&0\\1&0&0&0&1\\1&-1&-1&-1&-1\end{pmatrix}
\end{equation*}
The kernel of the transpose of this matrix gives the singly-graded degrees $(-5,1,1,1,1,1)$.  We will give a more detailed example in section \ref{ss:qeg}.

\subsection{The compactified moduli space}  \label{ss:mod}

\def\pervec{\boldsymbol{\Pi}}

We now build a coarse moduli space $\cM$ for the complex structures
given by polynomial deformations of $Y_\Upsilon$. By mirror symmetry,
this is gives the toric part of stringy moduli space of complexified
K\"ahler forms on $X_\Sigma$. 

$\cM$ is obtained by varying the coefficients $b_\alpha$ in 
(\ref{eq:Wb}). There is a $(\C^*)^d$-action given by rescalings of the
homogeneous coordinates. Thus $\cM$ appears as an open set in a toric
variety. As is well known, the appropriate toric variety is that given
by the {\em secondary fan\/} of regular simplicial decompositions of the pointset
$\cA$ \cite{GKZ:book,AGM:II}. That is,
\begin{equation*}
  \cM \subset S_\cA,
\end{equation*}
is an open subset where $S_\cA$ is the toric variety given by the
secondary polytope of the pointset $\cA$. $S_\cA$ is a
compactification of $T_\Cl\cong(\C^*)^r$, where $\Cl$
appears in (\ref{eq:M-D}).  This gives a natural compactification of
$\cM$ and we put $\cMcp=S_\cA$.

There are two distinct contributions to the pointset $\cMcp-\cM$ of
``bad theories'':
\begin{enumerate}
\item Part of the toric compactification $\cMcp-T_\Cl$, that is, the union
  of the toric divisors. Note that not all the points in
  $\cMcp-T_\Cl$ are necessarily bad as we discuss further below.
\item The discriminant $\Delta_W$. This is the ``principal
  $A$-determinant'' $E_A$ of \cite[Ch.~10]{GKZ:book}.\footnote{We have
    an unfortunate conflict of terminology here. We will use the term
    ``discriminant'' to refer to GKZ's ``A-determinant'', and
    ``A-discriminant'' to refer to GKZ's ``A-discriminant''.}
Aside from the
  toric divisors, it is where $Y_\Upsilon$ becomes singular. This
  was shown in the \GLSM\ language in \cite{MP:inst}.
\end{enumerate}

In the example of the quintic, there are two regular simplicial decompositions of
$\cA$, the secondary fan is one-dimensional, and $\cMcp\cong\P^1$. An
affine coordinate on this space is given from (\ref{eq:Wb}) as
\begin{equation*}
  z = \frac{b_1b_2b_3b_4b_5}{b_0^5}.
\end{equation*}
The toric divisors are $z=0$, the large radius limit, and $z=\infty$,
the \LG\ point. The discriminant $\Delta$ is the ``conifold'' point
$z=-5^{-5}$. Note that the \LG\ theory is perfectly well-defined so
$z=\infty$ is in $\cM$, while the large radius limit $z=0$ is not.

Define $\cM_0$ as the moduli space with these ``good'' points on
toric divisors removed. That is,
\begin{equation*}
  \cM_0 = \cM \cap T_\Cl.
\end{equation*}
Thus, in the quintic case, $\cM_0$ is a $\P^1$ with 3 points removed.

The object of study in monodromy is $\pi_1(\cM_0)$. Typically this is
a very complicated group, but there are some natural elements thanks
to the results of GKZ \cite{GKZ:book} as we now describe.

Each maximal cone in the secondary fan (or vertex of the secondary
polytope) is associated with a simplicial decomposition $\Sigma$ of $\cA$
and a ``phase'' in the sense of \cite{W:phase}. It is also naturally
associated, in the usual toric geometric sense, with a point in
$\cMcp=S_\cA$. We call this point the ``limit point'', $P_\Sigma$, of
the phase. As is standard in the construction of toric varieties,
there is a neighborhood $U_\Sigma$ of $P_\Sigma$ isomorphic to an affine
toric variety associated to the cone given by $\Sigma$.

One of the key results of GKZ \cite{GKZ:book} is that the discriminant
$\Delta_W$ stays ``well away'' from the limit points $P_\Sigma$ in the
following sense. The Newton polytope of $\Delta_W$ is precisely the secondary
polytope. Furthermore, as we move towards a limit point $P_\Sigma$, a
particular monomial in $\Delta_W$ corresponding to this vertex of
the convex hull of the secondary polytope acquires a larger and
larger absolute value and dominates all other monomials. Thus there is
some contractible neighborhood $V_\Sigma \subset U_\Sigma$ of $P_\Sigma$
which does not intersect $\Delta_W$. Furthermore, removing the toric
divisors we have
\[ \pi_1(V_\Sigma\cap\cM_0) \cong \pi_1(T_\Cl) = T_\Cl \cong \Z^r. \]

Thus, associated to each limit point $P_\Sigma$ we have $r$ commuting
elements $\gamma_{\Sigma,1},\ldots\gamma_{\Sigma,r}$ of $\pi_1(\cM_0)$. Note
that the geometry of $\cM_0$ away from $V_\Sigma$ may induce further
relations between these elements.

\subsection{Perestroika}  \label{ss:peres}

An often-used trick in analyzing the secondary polytope is to restrict
attention to the one-dimensional edges, i.e., the codimension one walls of
the secondary fan. Let $\Sigma_+$ and $\Sigma_-$ be two simplicial decompositions of $\cA$ at
the vertices of the secondary fan joined by such an edge. Going between
such triangulations was called a ``perestroika'' in \cite{GKZ:peres},
or a ``modification'' in \cite{GKZ:book}.

A codimension one wall in the secondary fan lies in a hyperplane
dual to a primitive vector $\pervec\in L$, defined up to a sign which we will eventually fix. This gives a ``circuit'', i.e., a collection of points that are affinely dependent but any subset is affinely
independent. The image of $\pervec\in L$ under $Q^t$ in
(\ref{eq:LZN}) gives the affine relation between the points.

Torically, a codimension one wall separating two maximal cones
corresponds to a toric $\P^1$ in the toric variety. That is, a
perestroika is associated to a toric rational
curve $\Theta\subset\cMcp$.
\begin{prop}
  The discriminant intersects $\Theta$ at a single point. Since $L$ is the
  character lattice of $T_\Cl$, the vector $\pervec\in L$ determines
  an affine coordinate $z_{\pervec}$ on $\Theta$. If
  $Q^t\pervec=(q_1,\ldots,q_n)$, then the discriminant $\Delta_W$
  meets $\Theta$ at
\begin{equation}
z_{\pervec} = \prod_{i=1}^n (q_i)^{q_i}. \label{eq:zhit}
\end{equation}
\end{prop}
\begin{proof}
This is proven using various results from \cite{GKZ:book} as follows. The
perestroika is associated to an edge $F$ of the secondary polytope. 
This is is associated to a ``polyhedral subdivision'' $\{(Q_i,\cA_i)\}$
in the sense of chapter 7. Each $(Q_i,\cA_i)$ is a simplex (with no
interior points) except for one --- the one associated to the
circuit. By
theorem $1.12'$ in chapter 10, we can restrict the discriminant to $\Theta$
to get a discriminant
\begin{equation*}
  \Delta_W\|_F = \prod_ic_i\Delta_{W,i}^{p_i},
\end{equation*}
where $\Delta_{W,i}$ is the discriminant associated to $(Q_i,\cA_i)$
and $c_i,p_i$ are numbers of no significance to this discussion. Each
$\Delta_{W,i}$ associated to a simplex is a constant. The
$(Q_i,\cA_i)$ associated to the circuit has two triangulations. In
this simple case, the secondary fan is one-dimensional and thus
$\Delta_{W,i}$ is a function of a single variable. Now use theorem 1.2
of chapter 10 to write $\Delta_{W,i}$ as a product of GKZ
``A-discriminants''. For each of these factors use theorem 3.3 of
chapter 9 to show that each A-discriminant is a point. Using the
affine coordinate $z_{\pervec}$, equation (3.3) of chapter 9
yields that this point is given by (\ref{eq:zhit}). Note that the
discriminant has many components but we see that the components that
meet $\Theta$ all meet it at the same point.
\end{proof}

We wish to restrict the monodromy question to such rational curves.
Unfortunately this does not quite make since since, unless $r=1$,
$\Theta\cap\cM_0$ is empty -- that is, $\Theta$ is contained in one or more of the toric divisors. To fix this we need to deform $\Theta$ a little, keeping the limit points fixed, but so that most of it lies in $\cM_0$. This requires a choice to be made. Note that, depending on the normal bundle of $\Theta$, we may require non-holomorphic deformations to achieve this, so we think of $\Theta$ as a 2-sphere rather than a rational curve. For example, consider $\P^1\times\C$ to have affine coordinates $(x,y)$. The rational curve at $y=0$ can be deformed to a 2-sphere
\begin{equation*}
  y = \begin{cases}\epsilon x &|x|\leq R\\
     \frac{\epsilon R^2}{\overline x}&|x|\geq R,\end{cases}
\end{equation*}
which still passes through $(0,0)$ and $(\infty,0)$.

The basic idea is that there are essentially 3 points of $\Theta$ around
which there is monodromy. Two of them are given by the limit
points of the two phases joined by the perestroika. The third point\footnote{When we deform $\Theta$ this may split into several points, but in that case we consider monodromy around all of them at once.}
comes from the discriminant $\Delta_W$ hitting $\Theta$ at $P_\Delta$ given by (\ref{eq:zhit}).

\subsection{Monodromy around phase limits}  \label{ss:mlim}

In a phase corresponding to a large radius \CY, the monodromy around
the large radius limit is fully understood. This corresponds to
shifting the $B$-field by an element of $H^2(X,\Z)$. The action on the
derived category is tensoring by a line bundle. In our toric
setting this amounts to the following.  A loop near the large radius limit stays well away
from the discriminant as we saw in section \ref{ss:mod}, hence corresponds to an element of $\pi_1(T_\Cl)=\Cl$. Monodromy around a
loop corresponding to $\boldsymbol{\rho}\in \Cl$ induces the autoequivalence
\begin{equation*}
  -\otimes\O_X(-\boldsymbol{\rho}),
\end{equation*}
on $\DC(X)$.

\begin{figure}
\begin{center}
\begin{tikzpicture}[scale=1.0]
\draw (0,0) ellipse (2.5 and 3);
\filldraw (0,2.2) circle (0.07);
\draw (-0.6,2.2) node {$\scriptstyle z=0$};
\draw (1,2.1) node {$\scriptstyle -\otimes S(-\boldsymbol{\rho})$};
\draw (0.85,-2.1) node {$\scriptstyle -\otimes S(\boldsymbol{\rho})$};
\draw (-0.6,-2.2) node {$\scriptstyle z=\infty$};
\draw [<-] (0,2.4) arc (90:-60:0.2);
\draw (0,2.4) arc (90:240:0.2);
\draw [<-] (0,-2.4) arc (270:120:0.2);
\draw (0,-2.4) arc (-90:60:0.2);
\draw (-0.1,2.027) -- (-0.1,-2.027);
\draw (0.1,2.027) -- (0.1,-2.027);
\draw [->] (-0.1,1.0) -- (-0.1,0.8);
\filldraw (-0.1,1.0) circle (0.03);
\draw (-0.1,1.0) node[anchor=east]{\tiny Start};
\filldraw (0,-2.2) circle (0.07);
\filldraw (1.3,0.3) circle (0.07);
\draw (1.4,0.3) node[anchor=south west]{$P_\Delta$};
\draw (1.3,0.3) circle (0.2);
\draw [->] (1.3,0.5) -- (1.3001,0.5);
\draw[dashed] (-2.5,0) .. controls (-1.5,0.5) and (1.5,0.5)
.. (2.5,0);
\draw (-1.7,0.7) node {$\Sigma_+$};
\draw (-1.7,-0.3) node {$\Sigma_-$};
\end{tikzpicture}
\end{center} \caption{Monodromy in $\Theta$} \label{fig:mon1}
\end{figure}
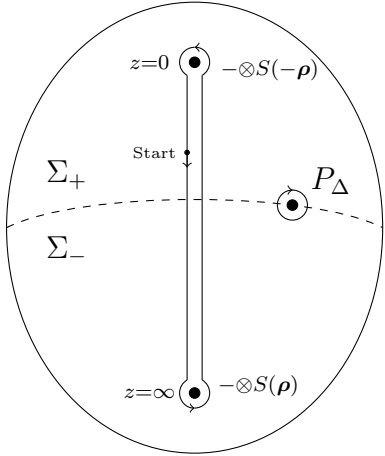

One of the key ideas in \cite{HHP:linphase} is to extend this to
monodromy around {\em all\/} limit points in $\cMcp$. In each phase the D-brane category is written as a quotient of the derived category of graded $S$-modules and the monodromy is $-\otimes S(-\boldsymbol{\rho})$. Thus we have a simple description of the monodromy action on the derived category
around all the limit points. The subtlety is that in each phase the quotient is different. This has been described in \cite{HHP:linphase,me:toricD,BFK:mf,MR2795327}; we also refer to \cite{me:hybridm,HW:unipot} for applications of this idea.\footnote{Indeed, \cite{me:hybridm,HW:unipot} may applied directly to find the monodromy for the limit points. However, the spherical functor approach we use gives a clearer picture of exactly which D-branes are \fluffy.}

Consider monodromy on a sphere with 3 points deleted as shown in figure
\ref{fig:mon1}. $\Theta$ determines a primitive vector $\pervec\in L$.
We fix the sign of $\pervec$ by saying that a vector
$\mathbf{v}$ in the interior of the maximal cone corresponding to
$\Sigma_+$ in the secondary fan satisfies $\langle
\pervec,\mathbf{v}\rangle>0$. 

Next we choose a vector $\boldsymbol{\rho}\in\Cl$ such that $\langle
\pervec, \boldsymbol{\rho}\rangle=1$. This choice of $\boldsymbol{\rho}$
corresponds to the choice of deformation of $\Theta$ at the end of section \ref{ss:peres}. The monodromies around $z=0$ and $z=\infty$ are then $-\otimes S(-\boldsymbol{\rho})$ and $-\otimes S(\boldsymbol{\rho})$ respectively. This in turn determines the monodromy around the discriminant point $P_\Delta$ as we will see in section \ref{ss:wallm}.

The main purpose of this paper is to compute the monodromy around $z=\infty$.  We have just said that this is given by $-\otimes S(\boldsymbol{\rho})$, but this form of the monodromy is too opaque for our purposes. First we will use this picture to find the monodromy around $P_\Delta$ and then use spherical functors to reformulate the monodromy around $z=\infty$.

\subsection{A Cubic Example} \label{ss:qeg}

The picture above is best illustrated by an example less trivial than
the quintic threefold. Let the pointset $\cA$ be given by the rows of
\begin{equation*}
  A^t=\begin{pmatrix}1&0&0&0&0\\1&1&0&0&0\\1&0&1&0&0\\
    1&0&0&1&0\\1&0&0&0&1\\1&-1&0&0&0\\1&1&-1&-1&-1\end{pmatrix}.
\end{equation*}
Then the matrix of charges is
\begin{equation}
  Q=\begin{pmatrix}-3&-1&0&1&1&1&1\\-2&1&1&0&0&0&0\end{pmatrix},
             \label{eq:Q6}
\end{equation}
so the superpotential is $W_X=x_0f$, where
\begin{equation}
f = x_1^2(x_3^5+x_4^5+x_5^5+x_6^5) +
x_2^2(x_3^3+x_4^3+x_5^3+x_6^3). \label{eq:f1}
\end{equation}
Here we have suppressed all the coefficients and listed only those
monomials at the vertices of the convex hull of the Newton polytope, writing a ``Fermat'' superpotential as an abbreviation for a generic superpotential: what we really mean is
\[ f = x_1^2\,g_5(x_3, \dotsc, x_6) + x_1 x_2\,g_4(x_3, \dotsc, x_6) + x_2^2\,g_3(x_3, \dotsc, x_6), \]
where $g_5$, $g_4$, and $g_3$ are polynomials of homogeneous of degree 5, 4, and 3 respectively. In fact $X_\Sigma$ would be singular if (\ref{eq:f1}) were used literally. This will be our convention for the superpotential from now on.

There are four triangulations of $\cA$, giving the following secondary fan:
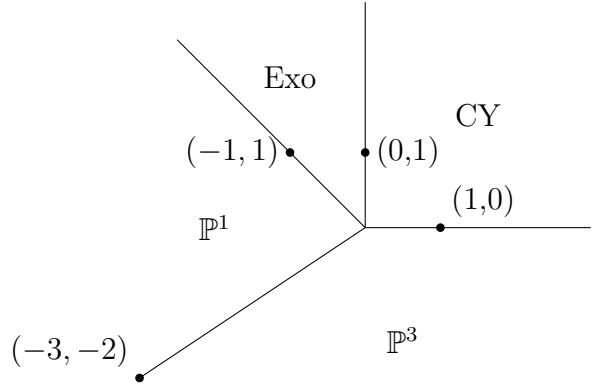
\begin{center}
\begin{tikzpicture}[scale=1.0]
  \draw (0,0) -- (3,0);  \draw(1.5,1.5) node {CY};
  \filldraw (1,0) circle (0.05) node[anchor=south west] {(1,0)};
  \draw (0,0) -- (-3,-2); \draw(0.5,-1.5) node {$\P^3$};
  \filldraw (-3,-2) circle (0.05) node[anchor=south east] {$(-3,-2)$};
  \draw (0,0) -- (-2.5,2.5); \draw(-2,0) node {$\P^1$};
  \filldraw (-1,1) circle (0.05) node[anchor=east] {$(-1,1)$};
  \draw (0,0) -- (0,3); \draw(-1,2) node {Exo};
  \filldraw (0,1) circle (0.05) node[anchor=west] {(0,1)};
\end{tikzpicture}
\end{center}
Two of the phases correspond to hybrid models fibered, with \LG\
fibre, over $\P^1$ and $\P^3$ respectively, and are not of interest
to us in this paper. The two other phases are
\begin{itemize}
\item A smooth \CY\ phase $\Sigma_+$, with Cox ideal
\[ B_+ = (x_1,x_2) \cap (x_3,x_4,x_5,x_6). \]
The toric variety $Z_+$ is the canonical line bundle of a $\P^1$-bundle $\P(\O(-1)\oplus\O)$ over $\P^3$.  Explicitly, $x_0$ is the fibre coordinate for the line bundle, $x_1$ and $x_2$ are coordinates on the $\P^1$, and $x_3, \dotsc, x_6$ are coordinates on the $\P^3$.
  The critical locus $X_+$ is given by $x_0=f=0$, hence is an anti-canonical divisor in the $\P^1$-bundle. It is a small resolution of a double cover of $\P^3$ branched over an octic with 60 nodes, so
\begin{equation*}
h^{1,1}(X_+)=2, \quad h^{2,1}(X_+) = 90.
\end{equation*}
\item An ``exoflop'' phase in the sense of \cite{AGM:II}, with Cox ideal
\begin{equation*}
  B_- = (x_0,x_1)\cap(x_1,x_2)\cap(x_2,x_3,x_4,x_5,x_6).
\end{equation*}
The critical locus $X_-$ has two components, as can be seen directly by
computing the minimal primes of the saturation
$(\Jac(W_X):B_-^\infty)$:
\begin{equation*}
  \{(x_0,f),(x_3,x_4,x_5,x_6)\}.
\end{equation*}
The first component, $x_0=f=0$, denoted $X^\sharp$, is a singular quintic
hypersurface in $\P^4$ (with homogeneous coordinates
$x_2, \dotsc, x_6$) obtained by setting $x_1=1$ in $f$ using one of the two
$\C^*$ actions. The other component is given by
$x_3=\dotsb=x_6=0$, so we can set $x_2=1$, leaving a $\P^1$ with
homogeneous coordinates $x_0,x_1$. This component is non-reduced and
corresponds to a hybrid theory with \LG\ fibre. The two components
intersect at the point $(0,1,1,0,0,0,0)$, which is the singular
point of $X^\sharp$.
\end{itemize}

We pass from the \CY\ phase to the exoflop phase by collapsing the
surface $E = \{ x_0=x_1=0 \} \subset X_+$ to a point.\footnote{We will discuss this in detail in the next section, but briefly it is because $(x_0,x_1)$ is the only minimal prime of $B_-$ that is not contained in $B_+$.}  One of the $\C^*$ actions can be used to
set $x_2=1$, so from (\ref{eq:f1}) we see that $E$ is a {\em cubic
surface\/}.
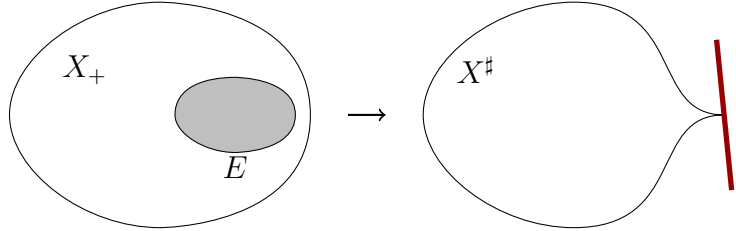
\begin{center}
\begin{tikzpicture}[scale=1.0]
\draw[yscale=0.75] (-5,0) .. controls (-5,1) and  (-4,2) .. (-3,2) ..
    controls (-3,2) and (-1,2) .. (-1,0) ..
    controls (-1,-2) and (-3,-2) .. (-3,-2) .. 
    controls (-4,-2) and (-5,-1) .. (-5,0);
\draw (-4,0.6) node {$X_+$};
\filldraw[yscale=0.25,xscale=0.4,xshift=-2cm,fill=gray!50] 
   (-5,0) .. controls (-5,1.5) and  (-4,2) .. (-3,2) ..
    controls (-3,2) and (-1,2) .. (-1,0) ..
    controls (-1,-2) and (-3,-2) .. (-3,-2) .. 
    controls (-4,-2) and (-5,-1) .. (-5,0);
\draw (-2,-0.7) node {$E$};
\draw [->,thick] (-0.5,0) -- (0,0);
\draw[yscale=0.75,xshift=5.5cm] (-5,0) .. controls (-5,1) and  (-4,2) .. (-3,2) ..
    controls (-1.5,2) and (-2,0) .. (-1,0) ..
    controls (-2,-0) and (-1.5,-2) .. (-3,-2) .. 
    controls (-4,-2) and (-5,-1) .. (-5,0);
\draw (1.2,0.6) node {$X^\sharp$};
\draw[xshift=5.5cm,color=red!60!black,line width=2pt] (-1.1,1) -- (-0.9,-1);
\end{tikzpicture}
\end{center}
Note that $X^\sharp$ can be deformed into a smooth quintic.  This extremal transition from $X_+$ to the quintic is of the form studied in
\cite{ACJM:srch}; such toric extremal transitions always involve
exoflops in some form. For the examples studied in this paper, the component ``sticking out'' of the \CY\ threefold may have dimension one or two.


\subsection{Wall monodromy}  \label{ss:wallm}

Given that we have a specific form for the monodromy around limit
points, we can understand the monodromy around the third
point in figure~\ref{fig:mon1}. That is, we can compute the monodromy around the discriminant $\Delta_W$ associated with the codimension one wall in the secondary fan separating two phases.

This question has been analyzed in \cite{HHP:linphase,me:toricD,BFK:mf,shipping} and the results confirm a form of monodromy described by Horja \cite{Horj:EZ}. The essential idea is that we have an autoequivalence related to collapsing a subspace $E\subset X$. The sheaf $\O_E$ is associated to modules annihilated by $B_-$ but not $B_+$. We repeat the analysis in this section for completeness and to fix notation.

Given that we are on a 2-sphere with 3 points removed, and as
figure~\ref{fig:mon1} shows, this wall monodromy is obviously given by
the composition of the two phase limit monodromies. We will follow the path
shown in figure~\ref{fig:mon1}. As described in
\cite{HHP:linphase,me:toricD,MR2795327}, to move between the phases,
i.e., the hemispheres in this figure, we need to express an
object in $\DC(X)$ in terms of a suitable tilting collection or ``window''.  It is
easiest to first address the problem in terms of the ambient toric
stack $Z_\Sigma$ using proposition \ref{prop:DZ}, and then pass to
$\DC(X_\Sigma)$ using matrix factorizations.

Assume that we have a tilting collection of the form
\begin{equation}
  S(\mathbf{q}_1), S(\mathbf{q}_2), \dotsc, S(\mathbf{q}_k), \label{eq:tilting_collection}
\end{equation}
such that
\begin{equation*}
  H^p_{B_\pm}(\mathbf{q}_i-\mathbf{q}_j)=0 \qquad \text{for $p>0$,}
\end{equation*}
where $H^p_{B_\pm}$ are local cohomology groups and the above is
simultaneously true for the irrelevant ideals $B_+$ and
$B_-$ of the two phases.  This means that corresponding line bundles on $Z_\pm$ satisfy
\begin{align*}
\Ext^p_{Z_\pm}(\O(\mathbf{q}_j), \O(\mathbf{q}_i)) &= 0 \qquad \text{for $p>0$} \\
\Hom_{Z_\pm}(\O(\mathbf{q}_j), \O(\mathbf{q}_i)) &= S_{\mathbf{q}_i - \mathbf{q}_j}.
\end{align*}
Furthermore we require that the line bundles $\O(\mathbf{q}_i)$ generate $\DC(Z_\pm)$ in both phases. This can be done for
all the examples in this paper, but there can be obstructions to its
happening in all phases \cite{me:toricD}.

To perform the monodromy in figure~\ref{fig:mon1} on an object in
$\DC(X)$ we do the following:
\begin{enumerate}
\item Write the object in ``canonical form'', i.e., as a complex of free modules which are sums of the modules \eqref{eq:tilting_collection}.
  This allows us to pass into the southern phase.
\item Apply the monodromy $-\otimes S(\boldsymbol{\rho})$.
\item If the object is longer in canonical form, apply mapping cones to
  or from trivial objects until we are back in canonical form so that
  we can pass back into the northern hemisphere.
\item Apply $-\otimes S(\boldsymbol{-\rho})$.
\end{enumerate}

An object is ``trivial'' in step 3 iff it is annihilated by a power of
$B_-$, so we now analyze which objects those are.  Since $B_\pm$ are square-free monomial ideals they are radical, hence can be written as the intersection of their minimal primes:
\begin{equation}
\begin{split}
B_+ &= \frak{m}_1\cap\frak{m}_2\cap\ldots\cap\frak{m}_s\\
B_- &= \frak{m}_1'\cap\frak{m}_2'\cap\ldots\cap\frak{m}_{s'}'.
\end{split} \label{eq:Bpm}
\end{equation}
The upshot will be that we can order these so that $\frak{m}_1$ is the only minimal prime of $B_+$ not contained in $B_-$, and $\frak{m}_1'$ is the only minimal prime of $B_-$ not contained in $B_+$.  Then the ``trivial'' objects in step 3 will be generated by $S/\frak{m}_1'$ and its grade shifts.

Alexander duality tells us that the irrelevant ideal $B$ is dual to
the Stanley--Reisner ideal of a simplicial decomposition $\Sigma$
\cite{MS:combcomm}, and that the minimal primes $\mathfrak{m}_i$ are of
the form $(x_{i_1},x_{i_2},\ldots)$, where
$\{\alpha_{i_1},\alpha_{i_2},\ldots\}$ is a {\em minimal\/} set of
points that do not form the vertices of a simplex in $\Sigma$.  A perestroika is supported along a {\em circuit\/} $Z\subset\cA$ as
described in section 7.2C of \cite{GKZ:book}, and as we first discussed in section~\ref{ss:peres}. There is a decomposition
$Z=Z_+\cup Z_-$ given by the signs of the unique affine relationship
between the points in $Z$. Associated to this are two distinct
simplicial decompositions of $Z$ which extend to the simplicial decompositions $\Sigma_+$ and
$\Sigma_-$ of $\cA$ related by the perestroika. Again see \cite{GKZ:book}
for the details of this. The sets $Z_\pm$ are minimal non-faces of
$\Sigma_\pm$, and without loss of generality we can set
\begin{equation}
\begin{split}
  \frak{m}_1 &= (x_{i_1},x_{i_2},\ldots)\quad\textrm{for}\,\, Z_+=\{\alpha_{i_1},\alpha_{i_2},\ldots\}\\
  \frak{m}_1' &= (x_{j_1},x_{j_2},\ldots)\quad\textrm{for}\,\,
  Z_-=\{\alpha_{j_1},\alpha_{j_2},\ldots\}
\end{split} \label{eq:pmp}
\end{equation}
For the example in section~\ref{ss:qeg}, the wall given by the line $(0,1)$ in the secondary fan is dual to the
vector $\pervec=(1,0)\in L$, and this corresponds to a perestroika for
an affine relation given by the first row of $Q$, i.e.,
\begin{equation*}
  -3\alpha_0-\alpha_1+\alpha_3+\alpha_4+\alpha_5+\alpha_6=0.
\end{equation*}
Thus $\frak{m}_1' = (x_0,x_1)$ corresponds to the cubic surface $E \subset X_+$ that collapses in passing to the exoflop phase, and $\frak{m}_1 = (x_3,x_4,x_5,x_6)$ corresponds to the $\P^1$ component ``sticking out'' of $X_-$.

\begin{prop}
For two irrelevant ideals of the form (\ref{eq:Bpm}) related by a
perestroika supported on the first minimal primes (that is, we satisfy
(\ref{eq:pmp})) we have $\frak{m}_i\supset B_-$ 
and $\frak{m}_i'\supset B_+$ for any $i>1$.
\end{prop}
\begin{proof}
Consider a set of points
$\{\alpha_{i_1},\alpha_{i_2},\ldots\}$ which are not the vertices of a
simplex in $\Sigma_+$ but which are the vertices of a simplex in $\Sigma_-$.
Thus $x_{\alpha_{i_1}}x_{\alpha_{i_2}}\ldots$ lies in the
Stanley--Reisner ideal $B_+^\vee$ but not in $B_-^\vee$. From the
description of simplicial decompositions related by a perestroika in section 7.2C
of \cite{GKZ:book}, this must mean that a nonempty subset
$D\subset\{\alpha_{i_1},\alpha_{i_2},\ldots\}$ also lies in $Z_+$. If
such a set of points arose from a minimal prime $\frak{m}_i$ of $B_+$
with $i\neq1$ then $D$ must be a {\em proper\/} subset of
$\{\alpha_{i_1},\alpha_{i_2},\ldots\}$. This is a contraction with
$\frak{m}_i$ being minimal since the subset $D$ is also a non-face of
$\Sigma_+$. So if $\{\alpha_{i_1},\alpha_{i_2},\ldots\}$ gives a minimal
non-face of $\Sigma_+$ it must also be a non-face $\Sigma_-$. The proposition is
the Alexander dual of this statement.
\end{proof}

One should note that a minimal
non-face of $\Sigma_+$ need not be a {\em minimal\/} non-face of $\Sigma_-$. 
Therefore, the minimal primes $\frak{m}_2,\ldots,\frak{m}_s$ and
$\frak{m}_2',\ldots,\frak{m}_{s'}'$ in (\ref{eq:Bpm}) need not
coincide. Indeed we saw this in the example in section \ref{ss:qeg}.

\begin{prop}
Let $B=\frak{m}_1 \cap \dotsb \cap \frak{m}_s$ be a decomposition of the irrelevant ideal
into minimal primes. If $M$ is a finitely-generated $S$-module
annihilated by some power of $B$ then $M$ is in the full triangulated
subcategory of $\DC(\gr(S))$ generated by $S/\frak{m}_1, \dotsc, S/\frak{m}_s$ and their grade shifts.
\end{prop}
\begin{proof}
Assume that $M$ is annihilated by $B^N$. Consider the
short exact sequence
\begin{equation*}
\xymatrix@1@M=2mm{
0\ar[r]&BM\ar[r]&M\ar[r]&M'\ar[r]&0.
}
\end{equation*}
Then $M'$ is annihilated by $B$ and $BM$ is annihilated by
$B^{N-1}$; thus by induction it is enough to prove the proposition for modules annihilated by $B$. Suppose $M$ is such a module, and consider the exact sequence
\begin{equation*}
\xymatrix@1@M=2mm{
0\ar[r]&\frak{m}_1M\ar[r]&M\ar[r]&M_1\ar[r]&0.
}
\end{equation*}
The module $M_1$ is annihilated by $\frak{m_1}$ and so is an
$S/\frak{m}_1$-module. Since $\frak{m}_1$ is an ideal simply of the form
$(x_{i_1},x_{i_2},\ldots)$, the ring $S/\frak{m}_1$ has finite global
dimension and thus $M_1$ has a finite free resolution in terms
of $S /\frak{m}_1$ and its grade shifts. To analyze the remaining
part, $\frak{m}_1M$, repeat the process
with
\begin{equation*}
\xymatrix@1@M=2mm{
0\ar[r]&\frak{m}_2\frak{m}_1M\ar[r]&\frak{m}_1M\ar[r]&M_2\ar[r]&0,
}
\end{equation*}
and resolve $M_2$ by a finite free resolution in
$S/\frak{m}_2$. Repeating the process for all the $\frak{m}_i$'s
finally annihilates $M$ for the left term in the short exact sequence
completing the proof. 
\end{proof}

Now the module $S/\frak{m}_i$ is obviously annihilated by $B_+$, but it is also annihilated by $B_-$ if $i>1$.  Thus the effect of
steps 1--4 above to applying mapping cones to and from $S/\frak{m}_1'$ and its grade shifts.

So far we have discussed D-branes on the ambient $Z_\Sigma$. To pass
to $\DC(X_\Sigma)$ we consider the category of matrix factorizations of
$W$. As discussed in \cite{HHP:linphase,me:toricD,MR2795327,BFK:mf},
we can again use the notion of tilting collection or window. That
is, we carry D-branes between the phases using matrix factorizations
involving free modules whose summands are in the tilting collection.
The result again is that the only effect of going around the loop is
to applying mapping cones to and from the matrix factorization
corresponding to $S/\frak{m}_1'$ and its grade shifts.

In every example in this paper, the wall monodromy turns out to be an EZ-spherical twist in the sense of Horja \cite{Horj:EZ}.\footnote{It has recently been proved that this is a general feature of ``window shifts'' \cite[\S3.2]{shipping}.} That is, moving
to the wall of a large radius \CY\ phase is equivalent to moving to a
wall of the K\"ahler cone where some subset $E\subset X_\Sigma$ collapses down to a smaller variety $Z$, with maps
\begin{equation}
\xymatrix{
E\, \ar[d]^q\ar@{^{(}->}[r]^i &X_\Sigma\\
Z
} \label{eq:EZ}
\end{equation}
and the monodromy gives an autoequivalence
\begin{equation}
  \mathsf{B} \mapsto \Cone(i_*q^*q_*i^!\mathsf{B}\to\mathsf{B}).
        \label{eq:EZm}
\end{equation}
Up to shifts depending on the tilting collection used, the object
$S/\frak{m}_1'$ corresponds to a matrix factorization giving the
sheaf $\O_E$. 

In the case that $Z$ is a point, this reduces to the Seidel--Thomas
spherical twist \cite{ST:braid}
\begin{equation}
\mathsf{B} \mapsto \Cone(\RHom(\O_E,\mathsf{B})\otimes\O_E\to
 \mathsf{B}), \label{eq:ST}
\end{equation}
where $\O_E$ is a spherical object, i.e.,
\begin{equation*}
\Ext^i(\O_E,\O_E) = \begin{cases}\C&\quad\hbox{if $i=0$ or 3}\\
0&\quad\hbox{otherwise.}\end{cases}
\end{equation*}

For the example in section \ref{ss:qeg} we can use a tilting collection
\begin{equation*}
  T = \bigoplus_{\substack{0\leq p<4\\0\leq q<2}} S(p,q).
\end{equation*}
The perestroika between the \CY\ phase and the exoflop phase is given
by $\pervec=(1,0)$, and $\frak{m}_1'=(x_0,x_1)$. Note that $\coker(x_0)$
always corresponds (up to shifts) to $\O_X$ given as the obvious
matrix factorization
\begin{equation*}
  \xymatrix@R=0mm@C=20mm{
     S(-3,-2) \ar@<1mm>[r]^-f & S\ar@<1mm>[l]^-{x_0}.}
\end{equation*}
Thus $S/(x_0,x_1)$ can be written as the cokernel of
\begin{equation*}
\xymatrix@1{\O_X(1,-1)\ar[r]^-{x_1}&\O_X,}
\end{equation*}
which is $\O_E$, where $E$ is the cubic surface. In this
case, the surface contracts to a point corresponding to $Z$
in (\ref{eq:EZ}). So the monodromy corresponding to the wall
is given by a Seidel--Thomas spherical twist on $\O_E$. 

We should note that it was the choice of sign convention $\langle
\pervec,\boldsymbol{\rho}\rangle=+1$ that leads to an EZ-transformation rather than its inverse. In the case of the cubic surface, a basis was chosen for the matrix $Q$ in (\ref{eq:Q6}) so that $\boldsymbol{\rho}=(1,0)^t$ and $\pervec=(1,0)$. This will always be the case in further examples.
We also note for future use that in all cases, $\boldsymbol{\rho}$ corresponds to an effective divisor class.

\section{Spherical Functors}  \label{s:sphF}

In the last section we encountered the spherical twists of Seidel--Thomas and Horja.  In this section we review a generalization due to Rouquier \cite{MR2258045} and Anno \cite{Anno:sph}; for a more thorough review see \cite[\S1]{nick}.  The essential idea will be that a semiorthogonal decomposition of $\DC(Z)$, where $Z$ is as in \eqref{eq:EZ}, yields a factorization of the autoequivalence \eqref{eq:EZm}.

\begin{definition} \label{def:sph}
Let $\A$ and $\B$ be triangulated categories admitting Serre functors $S_\A$ and $S_\B$, and $F: \A \to \B$ an exact functor with right adjoint $R: \B \to \A$.  The {\bf cotwist} $C$ and {\bf twist} $T$ associated to $F$ are the cones on the unit and counit of the adjunction:\footnote{There are several ways to make these cones of functors rigorous.  One is to require that $\A$ and $\B$ be admissible subcategories of $\DC(X)$ and $\DC(Y)$ for smooth projective varieties $X$ and $Y$, and $F$ be induced by an object of $\DC(X \times Y)$.  Another is to choose DG-enhancements of $\A$ and $\B$.}
\begin{align*}
C &= \Cone(1 \xrightarrow\eta RF) & T &= \Cone(FR \xrightarrow\epsilon 1). 
\end{align*}
The functor $F$ is called {\bf spherical} if:
\begin{enumerate}
\item \label{C_is_equiv}
$C$ is an equivalence, and
\item \label{CY_condition}
$S_\B F C \cong F S_\A.$\footnote{Many authors require that a certain natural map $F S_\A \to S_\B F C$ be an isomorphism, but in fact any isomorphism will do \cite[\S1.5]{nick}.}
\end{enumerate}
\end{definition}

\begin{theorem}[Rouquier, Anno]
If $F: \A \to \B$ is spherical then the twist $T: \B \to \B$ is an equivalence.
\end{theorem}

\begin{example}[Seidel--Thomas] \label{ex:ST}
In section \ref{ss:wallm} we had a spherical object $\O_E$ on a \CY\ threefold $X$.  To put it in this framework, take $\A = \DC(\text{point})$, i.e., the category of graded vector spaces, $\B = \DC(X)$, and $F = \O_E \otimes -$.  Then $R = \RHom(\O_E, -)$, and the spherical condition $\RHom(\O_E,\O_E) = \C \oplus \C[-3]$ is equivalent to saying that the cotwist $C = [-3]$, so conditions \ref{C_is_equiv} and \ref{CY_condition} above are satisfied.  The twist $T$ is exactly \eqref{eq:ST}.
\end{example}

\begin{example}[Horja]
In the more general setup of \eqref{eq:EZ} and \eqref{eq:EZm}, we take $F = i_* q^*: \DC(Z) \to \DC(X)$; some of the hypotheses were suppressed in our earlier discussion, but the twist $T$ is exactly \eqref{eq:EZm}.
\end{example}

\begin{example}[inclusion of a divisor] \label{eg:divisor}
This example is simple but will prove very useful.  Let $i: E \hookrightarrow X$ be the inclusion of a divisor in a \CY\ threefold, and $F = i_*$.  (This can be seen as a Horja twist in which $E = Z$ and $q$ is the identity.)  Then we find that the cotwist $C = \O_E(E)[-1] \otimes -$, which equals $S_E[-3]$ by the adjunction formula, so condition \ref{CY_condition} above is satisfied, and the twist $T = \O_X(E) \otimes -$.
\end{example}

Of course we did not need any fancy technology to tell us that this last $T$ is an equivalence, but the example becomes interesting when paired with the following fact:

\begin{theorem}[Kuznetsov, unpublished] \label{thm:twist_factorization}
Let $F: \A \to \B$ be a spherical functor with cotwist $C = S_\A[-k]$ for some $k \in \Z$,\footnote{This is typically the case when $\B$ is \CY\ of dimension $k$; compare condition \ref{CY_condition} of definition \ref{def:sph}.} and twist $T$.  If $\A$ admits a semi-orthogonal decomposition
\begin{equation*}
\A = \langle \A_1, \A_2, \dotsc, \A_n \rangle
\end{equation*}
then $F_i := F|_{\A_i} : \A_i \to \B$ is spherical with cotwist $C_i = S_{\A_i}[-k]$ for each $i$, and the twists $T_i$ satisfy
\begin{equation}
T_1 T_2 \dotsm T_n = T. \label{eq:twist_factorization}
\end{equation}
\end{theorem}
\begin{proof}[Sketch proof]
By induction we can assume that $n=2$.  The statement that $F_i$ is spherical with cotwist $S_{\A_i}[-k]$ is \cite[\S1.2, Prop.]{nick}, so we need only prove \eqref{eq:twist_factorization}.  We give a rough argument making free use of double cones, which are really only legitimate in a DG-enhancement, and then indicate an alternative argument that stays in triangulated categories.  For a more elaborate proof see \cite[Thm.~4.13]{shipping}.

Let $I_i: \A_i \to \A$ be the inclusions and $I_i^l, I_i^r: \A \to \A_i$ their left and right adjoints.  Then $F_i = F I_i$, so the twists $T_i$ are the cones
\begin{equation*}
T_i = \Cone(F I_i I_i^r R \xrightarrow{\epsilon_i} 1),
\end{equation*}
and their composition is the double cone
\begin{equation} \label{eq:double_cone_1}
T_1 T_2 = \Cone(F I_1 I_1^r R  F I_2 I_2^r R
\xrightarrow{\left( \begin{smallmatrix} \epsilon_2 \\ \epsilon_1 \end{smallmatrix} \right)}
F I_1 I_1^r R \oplus F I_2 I_2^r R
\xrightarrow{\left( \begin{smallmatrix} \epsilon_1 & -\epsilon_2 \end{smallmatrix} \right)}
1).
\end{equation}

To simplify the first term of \eqref{eq:double_cone_1}, take the exact triangle
\begin{equation*}
1 \to RF \to C,
\end{equation*}
apply $I_1^r$ on the left and $I_2$ on the right to get
\begin{equation*}
I_1^r I_2 \to I_1^r RF I_2 \to I_1^r C I_2,
\end{equation*}
and observe that
\begin{equation*}
I_1^r C I_2 = I_1^r S_\A I_2 [-k] = S_{\A_1} I_1^l I_2 [-k] = 0,
\end{equation*}
where in the last step we have $I_1^l I_2 = 0$ because $\Hom(\A_2,\A_1) = 0$.  Thus $I_1^r RF I_2 = I_1^r I_2$, so \eqref{eq:double_cone_1} becomes
\begin{equation} \label{eq:double_cone_2}
T_1 T_2 = \Cone(F I_1 I_1^r I_2 I_2^r R \to F I_1 I_1^r R \oplus F I_2 I_2^r R \to 1).
\end{equation}

Next observe that the projections into $\A_i^\perp$ are
\begin{equation*}
\Cone(I_i I_i^r \to 1),
\end{equation*}
and their composition is on the one hand zero (since the $\A_i$ generate $\A$) and on the other hand the double cone
\begin{equation*}
\Cone(I_1 I_1^r I_2 I_2^r \to I_1 I_1^r \oplus I_2 I_2^r \to 1).
\end{equation*}
Thus we have
\begin{equation*}
\Cone(I_1 I_1^r I_2 I_2^r \to I_1 I_1^r \oplus I_2 I_2^r) = 1,
\end{equation*}
so \eqref{eq:double_cone_2} becomes
\begin{equation*}
T_1 T_2 = \Cone(FR \to 1) = T
\end{equation*}
as desired.

For the reader who is suspicious of double cones, we mention that it is also possible to argue using the exact triangle
\begin{equation*}
I_2 I_2^r \xrightarrow{\epsilon_2} 1 \xrightarrow{\eta_1} I_1 I_1^l
\end{equation*}
and a few applications of the octahedral axiom.
\end{proof}

In view of Definition \ref{def:fluffy}, we are interested in objects of $\B$ on which $T$ acts in a particularly simple way.  We will concentrate on objects of the form $F\a$, where $\a \in \A$.  The following generalizes the well-known fact that the spherical twist \eqref{eq:ST} sends $\O_E$ to $\O_E[-2]$:

\begin{prop} \label{prop:TF=FC}
If $F$ is a spherical functor with cotwist $C$ and twist $T$, then $TF \cong FC[1]$.
\end{prop}
\begin{proof}
This is covered in \cite[\S1.3]{nick}, but it is essential to the present paper and the proof is fun and easy, so we give the proof in full.  By standard category theory, the diagram
\begin{equation*} \xymatrix{
F \ar[r]^{F\eta} \ar@{=}[rd] & FRF \ar[d]^{\epsilon F} \\ & F
} \end{equation*}
commutes.  Extend this to
\begin{equation*} \xymatrix{
& TF[-1] \ar[d] \ar[rd] \\
F \ar[r]^{F\eta} \ar@{=}[rd] & FRF \ar[d]^{\epsilon F} \ar[r] & FC \\
& F
} \end{equation*}
where the row and column are exact.  Then by the octahedral axiom, the cone on the map $TF[-1] \to FC$ equals the cone on $\xymatrix{F \ar@{=}[r] & F}$, which is zero, so $TF[-1] \to FC$ is an isomorphism.
\end{proof}

Finally, in order to commute spherical twists past one another we will need the following fact, which follows easily from the definitions:
\begin{prop} \label{prop:PhiF}
If $F: \A \to \B$ is a spherical functor with cotwist $C$ and twist $T$ and $\Phi$ is any autoequivalence of $\B$, then $\Phi F$ is spherical with cotwist $C$ and twist $\Phi T \Phi^{-1}$.  Put another way, we have
\begin{flalign*}
&& \Phi T_F = T_{\Phi F} \Phi. && \qed
\end{flalign*}
\end{prop}

\section{\fluffy ness} \label{s:fluffy}

\subsection {Monodromy and Stability} \label{ss:stab}

D-branes stability is closely tied to the concept of monodromy
\cite{AD:Dstab}. As one moves in the moduli space $\cM$, the
stability condition changes. On following fixed objects in $\DC(X)$
around a loop in this, the stability at the start of the
loop may differ from the stability at the end.  The set of stable
objects form a nontriangulated full subcategory of $\DC(X)$.  The
automorphism of the derived category associated to a loop can be
viewed as that which induces an equivalence between the subcategories
of stable objects at the start and at the end of the loop.

In some cases the object remains stable for the full traverse of the
loop and the monodromy can be deduced by a more direct method. For
example, an object stable near the large radius limit can remain
stable for a shift in the $B$-field. Then we get the monodromy
described in section \ref{ss:mlim}. Another case, of central
importance to this paper, concerns loops around ``bad'' points where a
D-brane goes massless. 

A stable object $\a$ has a real number
$\xi(\a)$ associated to it, where
\begin{equation*}
  \xi(\a) = \frac1\pi\arg \cz(\a)\pmod2,
\end{equation*}
and $\cz(\a)$ is the central charge.  This number varies continuously (so long as the object is stable) as we move in $\cM$. There is also the relation
\begin{equation*}
  \xi(\a[n]) = \xi(\a)+n.
\end{equation*}
Now suppose we have a disk in $\cM$ parametrized by $y$ on which $\a$ is stable for $y\neq0$, but we have monodromy $\a\mapsto\a[2]$ around $y=0$. It follows from the above relations that $Z(\a)$ has a zero at $y=0$.

Mass is proportional to $|\cz(\a)|$. 
Conformal field theories are ``bad'' when stable D-branes become
massless \cite{Str:con}. So it is very natural to see monodromy like
this around the compactification divisor of $\cMcp$. This motivates the following:

\begin{definition} \label{def:fluffy}
Let $T$ be the autoequivalence associated to monodromy in a positive
direction around a small holomorphic disk in $\cMcp$ which intersects
the compactification divisor of $\cMcp$ at a single point $P$ in the
interior. We define an object $\a$ to be {\bf\fluffy}\
at $P$ if we have
\begin{equation*}
  T^q(\a) = \a[p].
\end{equation*}
for some positive integers $p$ and $q$. 
\end{definition}

If $y$ is again the coordinate of the disk
with $P$ given by $y=0$ then this is consistent with the central charge $\cz$ going like $y^{p/2q}$. Since the mass is proportional to $|\cz|$, the object becomes massless at $P$ (except in \LG\ theories as we discuss shortly).

Note that the spherical twist (\ref{eq:ST}) sends $\O_E$ to
$\O_E[-2]$. Given that the monodromy in figure \ref{fig:mon1} is in a
negative direction around $P$, this is consistent with $\cz$ varying
as $y$ for the D-brane $\O_E$. Indeed, as discussed in \cite{AKH:m0},
the EZ-transformations produce a whole class of objects for which $\cz$
varies as $y$ which are \fluffy\ with $(p,q)=(1,2)$.

It is important to note that $\xi$ is only defined for a {\em
  stable\/} object in $\DC(X)$. Indeed, there are objects whose K-theory
class would imply that $\cz=0$ at some point $P\in\cMcp$ but which
cannot possibly be massless D-branes. For example, let $X$ be the
canonical line bundle over $\P^2$ and let $\ell$ be a line on this
$\P^2$. The analysis in section 7.3.4 of \cite{me:TASI-D} shows that
the object $\O_\ell(-1)$ has $\cz=0$ at the orbifold point where the $\P^2$
is shrunk down to a point. Since the orbifold is a smooth CFT, this
object cannot be stable. But monodromy around this orbifold point
satisfies $T^3=\id$ on all objects in $\DC(X)$ and so $\O_\ell(-1)$
cannot be \fluffy.

A very interesting question is whether an object being \fluffy\ at $P$
is sufficient for it to be a stable massless D-brane. It is
extremely difficult to prove that an object is stable at a given point
in $\cMcp$, but it is known that at least some stable objects must
become massless if $P$ corresponds to a singular conformal field
theory.

Note that there are cases where $\cz$ goes as $z^{p/2q}\log(z)$, as we
will see later. Thus a stable massless D-brane is not necessarily
\fluffy. There is also a special case that needs some care where {\em every\/} object is $\Q$-massless, i.e., $T^q=[p]$ identically. This happens in \LG\ orbifold theories as is well-known for the quintic for example \cite{me:TASI-D}. The central charge $Z$ can be computed by computing periods
\begin{equation*}
   \varpi(\Gamma) = \int_{\Gamma}\Omega,
\end{equation*}
on 3-cycles $\Gamma$ in the mirror $Y$ of $X$. The central charge is then computed by $Z(\Gamma)=\kappa^{-\frac12}\varpi(\kappa)$, where
\begin{equation*}
   \kappa = \int_Y\overline{\Omega}\wedge\Omega.
\end{equation*}
For a \LG\ orbifold theory {\em all\/} of the periods $\varpi(\Gamma)$ and $\kappa$ are zero. The monodromy is thus related to an artifact of the normalization parameter $\kappa$ and none of the central charges actually vanish. Therefore, so long as {\em some\/} of the objects in $\DC(X)$ have nonzero periods and $Z$ as $y\to0$ we know $\kappa$ cannot be zero. In the cases of interest in this paper, we will keep part of $X$ at large radius which enforces this condition. It follows that not all objects will be $\Q$-massless, and that those which are have $Z\to0$.

\subsection{Fractional \CY\ Categories}  \label{ss:Kuz}

Using proposition \ref{prop:TF=FC}, we will obtain many \fluffy\ objects from spherical functors $F: \A \to \DC(X)$ where $X$ is a \CY\ threefold, the cotwist $C = S_\A[-3]$, and $\A$ is a ``fractional \CY\ category'' of dimension $p/q$, meaning that $S_\A^q = [p]$.  Indeed, under these hypotheses we have
\begin{equation*}
T^q(F\a) = F(C^q(\a))[q] = F(S_\A^q(\a))[-2q] = F(\a)[p-2q],
\end{equation*}
so the whole image of $F$ is \fluffy.  Our fractional \CY\ categories will be subcategories of derived categories of hypersurfaces in weighted projective space:

\begin{theorem}[Kuznetsov]  \label{thm:Kuz1}
Let $Y$ be a smooth hypersurface of degree $d$ in weighted projective space $\P^{n+1}_{\{w_0,w_1,\ldots,w_{n+1}\}}$.\footnote{In particular, $Y$ misses all the orbifold singularities in this weighted projective space.}  Let $k = \sum w_i - d$, so $\omega_Y = \O_Y(-k)$, and suppose that $k > 0$, so $Y$ is Fano.  Then the subcategory $\A \subset \DC(Y)$ defined by the semi-orthogonal decomposition
\begin{equation*}
  \DC(Y) = \langle \A,\O_Y,\O_Y(1),\ldots, \O_Y(k-1)\rangle,
\end{equation*}
has $S_\A^d = [nd-2k]$, i.e., $\A$ is fractional \CY\ of dimension $(nd-2k)/d$.
\end{theorem}
\begin{proof}[Sketch Proof]
\def\OJ{\mathsf{O}}
Consider the functor $\OJ:\DC(Y)\to\DC(Y)$ defined by
\begin{equation*}
  \OJ(\a) = L_{\O_Y}(\a\otimes \O_Y(1))[-1],
\end{equation*}
where $L_{\O_Y}: \DC(Y) \to \DC(Y)$ is left mutation past $\O_Y$:
\begin{equation} \label{eq:mutation}
L_{\O_Y}(\a) = \Cone(\RHom(\O_Y, \a) \otimes \O_Y \to \a).
\end{equation}
Then $\OJ$ preserves $\A$.  We have on the one hand $(\OJ|_\A)^k = S_\A^{-1}[n-k]$, and on the other hand $(\OJ|_\A)^d = [2-d]$.  Kuznetsov proves this for hypersurfaces in ordinary projective space in \cite[Lem.~4.1 and 4.2]{kuz:V14}, and the adaptation to weighted projective space is discussed in \cite[Rmk.~4.7]{CK:mon}.
\end{proof}

Let $Y$ be cut out by a degree-$d$ polynomial $W \in \C[x_0,\dotsc,x_{n+1}]$.  The category $\A$ in
the above proposition is exactly the category of graded
matrix factorizations of $W$ \cite{Orlov:mfc}, which is the category of D-branes for the $\Z_d$-orbifold of the \LG\ theory with superpotential $W$. It is interesting to note that the central charge of this theory, which is a measure of its dimension, is given by 
\[ \sum_{i=0}^{n+1} (1 - 2w_i/d) = (nd-2k)/d. \]

\section{Examples}   \label{s:app}

In section \ref{ss:wallm} we found the monodromy around $P_\Delta$ using our knowledge of the monodromy around the limit points $z=0$ and $z=\infty$.  We now switch tactics to obtain the monodromy around the limit point $z=\infty$ using the known monodromy around $z=0$ and $P_\Delta$. At first sight this seems perverse; the issue is that the way we described the monodromy around $z=\infty$ in section \ref{ss:mlim} as $-\otimes\O(\boldsymbol{\rho})$ is not very enlightening. The results will be more clear if we use the spherical functor language and theorem \ref{thm:twist_factorization}. We want to obtain $T_\infty$ in terms of the composition of $T_\Delta$ and $T_0$ as shown in figure \ref{fig:mon2}. Note that the orientation of the loops around $z=0$ are reversed with respect to figure \ref{fig:mon1}.

In each case the idea is to view $T_0$ as coming from an inclusion of a divisor
$E\hookrightarrow X$ as in example \ref{eg:divisor}. Then $T_0$ is split into $T_\Delta$ and $T_\infty$ according to theorem \ref{thm:twist_factorization} by doing a semiorthogonal decomposition of $\DC(E)$. This theorem also implies that all the cotwists involved satisfy
\begin{equation}
  C = S[-3].  \label{eq:cotwist}
\end{equation}

\begin{figure}
\begin{center}
\begin{tikzpicture}[scale=1.0]
\draw (0,0) ellipse (2.5 and 3);
\filldraw (0,2.2) circle (0.07);
\draw (-0.6,2.2) node {$\scriptstyle z=0$};
\draw (0,2.2) node[anchor=south west] {$T_0$};
\draw (-0.6,-2.2) node {$\scriptstyle z=\infty$};
\draw (0,2.4) arc (90:-25.6:0.2);
\draw [<-](0,2.4) arc (90:274.4:0.2);
\draw (1.3,0.1) arc (-90:-205.6:0.2);
\draw [<-](1.3,0.1) arc (-90:94.4:0.2);
\draw (0.01,2.003) -- (1.12,0.375);
\draw (0.178,2.12) -- (1.29,0.5);
\draw [->] (0.69,1.0) -- (0.5,1.285);
\filldraw (0.69,1.0) circle (0.03);
\draw (0.69,1.0) node[anchor=east]{\tiny Start};
\filldraw (0,-2.2) circle (0.07);
\draw (0,-2.2) circle (0.2);
\draw (0,-2.2) node[anchor=south west] {$T_\infty$};
\draw [->] (0,-2.0) -- (-0.001,-2.0);
\filldraw (1.3,0.3) circle (0.07);
\draw (1.4,0.3) node[anchor=south west]{$T_\Delta$};
\draw[dashed] (-2.5,0) .. controls (-1.5,0.5) and (1.5,0.5)
.. (2.5,0);
\draw (-1.7,0.7) node {$\Sigma_+$};
\draw (-1.7,-0.3) node {$\Sigma_-$};
\end{tikzpicture}
\end{center} \caption{Monodromy in $\Theta$} \label{fig:mon2}
\end{figure}
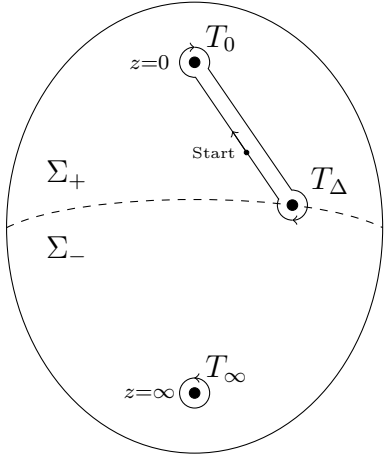

\subsection{Del Pezzo Surfaces} \label{ss:delP}

The first cases we consider are when the passage from phase $\Sigma_+$ to $\Sigma_-$ results in a del Pezzo surface $E \subset X = X_+$ collapsing to a point.  In each example we will have $\rank L = 2$; we will choose a basis of starting with the perestroika $\pervec \in L$, and let $D$ and $H$ be the dual basis of $\Cl$, which we identify with $\Pic(X)$.  Thus $T_0 = \O_X(D) \otimes -$ as we discussed at the end of section~\ref{ss:mlim}.  By the method of section~\ref{ss:wallm} we find that $T_\Delta = T_{\O_E}$, the Seidel--Thomas spherical twist around $\O_E$.  In each case we will have:
\begin{itemize}
\item $D$ and $H$ are effective divisor classes on $X$.
\item $E = H-mD$ as divisor classes on $X$, where $m$ is the ``index'' of the del Pezzo surface $E$, i.e., the greatest integer such that the canonical line bundle $\omega_E$ has an $m^\text{th}$ root: 3 for $\P^2$, 2 for $\P^1 \times \P^1$, and 1 otherwise.
\item $H|_E = 0$.  Note that if $E$ is to collapse while the rest of $X$ remains nonzero in size then there must be such a divisor class.
\end{itemize}

\subsubsection{$\P^2$}  \label{ss:P2}

Let $E=\P^2$. If $E\subset X$ is contracted to a point, it is
well-known that we obtain a local $\C^3/\Z_3$ orbifold singularity.  In example \ref{eg:divisor} we saw that the functor $i_*: \DC(E) \to \DC(X)$ is spherical, with twist
\begin{equation} \label{eq:T_DE}
T_{\DC(E)} = \O_X(E) \otimes -.
\end{equation}
By theorem \ref{thm:twist_factorization}, the semi-orthogonal decomposition
\begin{equation*}
  \DC(E) = \langle\O_E,\O_E(1),\O_E(2)\rangle
\end{equation*}
allows us to factor \eqref{eq:T_DE} as a product of Seidel--Thomas twists:
\begin{equation*}
  T_{\DC(E)} = T_{\O_E}\, T_{\O_E(1)}\, T_{\O_E(2)}.
\end{equation*}
With respect to figure \ref{fig:mon2} we have $T_0 = \O_X(D) \otimes -$, $T_\Delta = T_{\O_E}$, and $T_\infty=T_\Delta T_0$.  We have $\O_E(D) = \O_E(1)$, using the facts that $\O_X(E) = \O_X(H-3D)$, $\O_E(E) = \O_E(-3)$, and $\O_E(H) = \O_E$.  Thus using proposition \ref{prop:PhiF} we get
\begin{equation*}
\begin{split}
T_\infty^3
&= T_{\O_E}\, T_0\, T_{\O_E}\, T_0\, T_{\O_E}\, T_0 \\
&= T_{\O_E}\, T_{\O_E(1)}\, T_0^2\, T_{\O_E}\, T_0 \\
&= T_{\O_E}\, T_{\O_E(1)}\, T_{\O_E(2)}\, T_0^3 \\
&= T_{\DC(E)}\, T_0^3 \\
&= \O_X(E+3D) \otimes - \\
&= \O_X(H) \otimes -,
\end{split}
\end{equation*}
as is well-known --- see \cite{me:navi,me:TASI-D} for a less sophisticated derivation.  This is the monodromy around the large radius limit in the
direction we are ``ignoring''. The divisor $H$ is far away from $E$, so aside from global geometry issues we essentially have $T_\infty^3=1$. In particular there are no \fluffy\ objects. Note that $z=\infty$ is a
perfectly good orbifold conformal field theory where nothing
unpleasant happens and thus we know for sure that there are no massless D-branes at all.

To see this geometry in a global example, let $X$ be the resolution of the degree 18
hypersurface in weighted $\P^4_{\{9,6,1,1,1\}}$ as studied in
\cite{CFKM:II}. For the toric picture, the matrix $Q$ is given by
\begin{equation*}
\begin{array}{c|ccccccc}&x_0&x_1&x_2&x_3&x_4&x_5&x_6\\\hline
D&0&0&0&1&1&1&-3\\H&-6&3&2&0&0&0&1\end{array}.
\end{equation*}
The superpotential is $W=x_0f$, where (as always, suppressing coefficients and just giving monomials at the vertices of the Newton polytope)
\begin{equation*}
f = x_1^2 + x_2^3 + (x_3^{18}+x_4^{18}+x_5^{18})x_6^6.
\end{equation*}
The class of $E$ ($x_6=0$) is $H-3D$ as expected.

The model has four phases: \CY, orbifold, \LG, and hybrid, as explored in \cite{CDFKM:I,me:navi}. Here we are considering the perestroika between the \CY\ and orbifold phases. Note that in {\em none\/} of the four phase limits do we see massless D-branes. All the subsequent examples we consider have exoflop phases, and that is where we discover the \fluffy\ objects.

\subsubsection{$\P^1\times\P^1$}  \label{sss:P1P1}

Now suppose $E=\P^1\times\P^1$. This case was also considered in \cite{AM:delP,me:modD} but here we find a more elegant and precise solution to the question of massless D-branes.  As in the previous section, we apply example \ref{eg:divisor} and theorem \ref{thm:twist_factorization} to the spherical functor $i_*: \DC(E) \to \DC(X)$ and the exceptional collection
\[ \DC(E) = \langle \O_E(-1,0),\O_E(0,-1),\O_E,\O_E(1,1)\rangle, \]
to get
\[ T_{\DC(E)} = \O_X(E) \otimes -
= T_{\O_E(-1,0)}\, T_{\O_E(0,-1)}\, T_{\O_E}\, T_{\O_E(1,1)}. \]
Again we have $T_0 = \O_X(D) \otimes -$, $T_\Delta = T_{\O_E}$, and $T_\infty=T_\Delta\,T_0$, and now $\O_E(D) = \O_E(1,1)$, so
\begin{equation*}
\begin{split}
  T_\infty^2 &= T_{\O_E}\, T_0\, T_{\O_E}\, T_0 \\
       &= T_{\O_E}\, T_{\O_E(1,1)}\, T_0^2 \\
       &= T_{\O_E(0,-1)}^{-1}\, T_{\O_E(-1,0)}^{-1}\, T_{\DC(E)}\, T_0^2 \\
       &= T_{\O_E(0,-1)}^{-1}\, T_{\O_E(-1,0)}^{-1}\, (\O_X(E+2D) \otimes -)\\
       &= T_{\O_E(0,-1)}^{-1}\, T_{\O_E(-1,0)}^{-1}\, (\O_X(H) \otimes -).
\end{split}
\end{equation*}
Tensoring with $\O_X(H)$ has no effect on objects supported on $E$. The objects $\O_E(0,-1)$ and $\O_E(-1,0)$ are orthogonal to one another, so $T_{\O_E(0,-1)}$ acts trivially on $\O_E(-1,0)$ and vice versa, and $T_{\O_E(0,-1)}$ commutes with $T_{\O_E(-1,0)}$.  Thus $\O_E(0,-1)$ and $\O_E(-1,0)$ are \fluffy:
\begin{align*}
T_\infty^2 \O_E(0,-1) &= T_{\O_E(0,-1)}^{-1} \O_E(0,-1) = \O_E(0,-1)[2] \\
T_\infty^2 \O_E(-1,0) &= T_{\O_E(-1,0)}^{-1} \O_E(-1,0) = \O_E(-1,0)[2].
\end{align*}

To recast this in terms that will be convenient in the next section, write
\begin{align*}
\A &= \langle \O_E(-1,0),\O_E(0,-1) \rangle \\
   &= \langle \O_E, \O_E(1,1) \rangle^\perp \subset \DC(E).
\end{align*}
Observe that $\A \cong \DC(\text{2 points})$.  If $F: \A \to \DC(E)$ denotes the restriction of $i_*$ to $\A$, then by theorem \ref{thm:twist_factorization}, $F$ is spherical with cotwist $S_\A[-3] = [-3]$, and the spherical twist satisfies $T_\A = T_{\O_E(-1,0)}\, T_{\O_E(0,-1)}$.  Thus by proposition \ref{prop:TF=FC} we have
\[ T_\infty^2 F = T_\A^{-1} F = F C_\A^{-1} [-1] = F S_\A^{-1} [-1]  = F[2], \]
so any object in the image of $F$ is \fluffy.  This is perhaps unimpressive since any object in $\A$ is a sum of shifts of $\O_E(0,-1)$ and $\O_E(-1,0)$, but in the next section it will become more interesting.

To see this geometry in a global example, let the matrix $Q$ be given by
\begin{equation*}
\begin{array}{c|ccccccc}&x_0&x_1&x_2&x_3&x_4&x_5&x_6\\\hline
D&-2&1&-2&0&1&1&1\\H&-4&2&1&1&0&0&0\end{array}.
\end{equation*}
The superpotential is $W=x_0f$, where (using previous conventions)
\begin{equation}
f = x_1^2 + x_2^4(x_4^{10}+x_5^{10}+x_6^{10}) +
     x_3^4(x_4^2+x_5^2+x_6^2). \label{eq:quadric}
\end{equation}
The model has five phases, two of which have $X_\Sigma$ a smooth \CY\ threefold.
One of these, which we denote $\Sigma_+$, has irrelevant ideal $B=(x_2,x_3)(x_1,x_4,x_5,x_6)$.

The other phase of interest, $\Sigma_-$, has $B=(x_0,x_2)\cap(x_2,x_3)
\cap(x_1,x_3,x_4,x_5,x_6)$. This is an exoflop phase very similar to that of section \ref{ss:qeg}. The underlying topological space of $X_\Sigma$ has two irreducible components: a $\P^1$ with homogeneous coordinates $x_0,x_2$, and the singular hypersurface in
$\P^4_{5,2,1,1,1}$ given by $f=0$ with $x_2=1$. We thus have an extremal transition to the smooth degree 10 hypersurface in $\P^4_{5,2,1,1,1}$ from this exoflop phase.

Passing from $\Sigma_+$ to $\Sigma_-$, the surface $x_2=0$ (and thus we may set $x_3=1$) in $X_\Sigma$ contracts to a point. From (\ref{eq:quadric}), this is a quadric surface, i.e., $E=\P^1\times\P^1$.

Note that $h^{1,1}(X)=2$ and $h^{2,1}(X)=144$. If $X'$ is the smooth degree 10 hypersurface after the extremal transition, then $h^{1,1}=1$ and $h^{2,1}=145$. These numbers are consistent with a typical Higgs--Coulomb transition. That is
\begin{equation}
\hbox{\#Massless D-branes} = h^{1,1}(X)-h^{1,1}(X')-h^{2,1}(X)+h^{2,1}(X')=2,
\label{eq:Higgs}
\end{equation}
as noted in \cite{me:modD}.

This case therefore works very nicely. We expect two massless D-branes from the physics perspective, and we find two $\Q$-massless objects.

\subsubsection{$dP_6$}  \label{sss:dP6}

Now suppose $E=dP_6$, a del Pezzo surface of degree 3.  This can be written as a cubic surface in $\P^3$.  Again the exceptional collection
\[ \DC(E) = \langle\O_{\ell_1}(-1),\O_{\ell_2}(-1),\ldots,\O_{\ell_6}(-1),\O_E(-2h),\O_E(-h), \O_E \rangle \]
where $\ell_i$ are the exceptional lines of the blowup $E \to \P^2$ and $h$ is the pullback of a line in $\P^2$ that misses the blown-up points, allows us to factor $\O_X(E) \otimes -$ as a product of Seidel--Thomas twists.  But we prefer to group the first eight exceptional objects together and write
\[ \DC(E) = \langle \A, \O_E \rangle. \]
In other words, $\A = (\O_E)^\perp \subset \DC(E)$.  But note that while $\A$ has an exceptional collection of length 8, it is not the derived category of 8 points, since there are Exts from left to right in the exceptional collection.  Indeed, since $E$ is a cubic surface, by theorem \ref{thm:Kuz1} we have $S_\A^3 = [4]$.

Now with the same conventions as in previous sections, we get
\begin{equation*}
\begin{split}
  T_\infty &= T_{\O_E}\, T_0 \\
    &= T_\A^{-1}\, T_{\DC(E)}\, T_0 \\
    &= T_\A^{-1}(\O_X(E+D) \otimes -)\\
    &= T_\A^{-1}(\O_X(H) \otimes -).
\end{split}
\end{equation*}
Now we claim that for any object $\a \in \A$, the pushforward $i_* \a \in \DC(X)$ is \fluffy.  Again tensoring with $\O_X(H)$ has no effect.  Recalling that $T_\A$ is the spherical twist associated to the functor $F = i_*: \A \to \DC(X)$ we apply proposition \ref{prop:TF=FC} and (\ref{eq:cotwist}) to get
\begin{equation*}
T_\infty F = T_{\A}^{-1} F = F S_\A^{-1} [2],
\end{equation*}
and so
\begin{equation*}
T_\infty^3 F = F S_\A^{-3} [6] = F \mathsf [2].
\end{equation*}
So the objects in the image of $F$ are again $\Q$-massless. If we use a local coordinate $y=z^{-1}$ near the exoflop limit point we are claiming that the entire contents of the category $\A$ has a central charge going like $y^{\frac13}$ as $y\to0$. We appear to have a considerable number of massless D-branes, including
\begin{itemize}
\item $\O_c(-1)$ and $\O_E(-c)$ for any smooth rational curve $c \subset E$, e.g.,
\begin{itemize}
\item The famous 27 lines.
\item The preimage via the blow-up map $E \to \P^2$ of a line or conic in $\P^2$.
\item The proper transform of a cubic curve in $\P^2$ with a node at one of the 6 blown-up points.
\end{itemize}
\item $\cI_p$, the ideal sheaf of any point $p \in E$.
\end{itemize}
Note that $\O_E$ itself is {\em not\/} massless. This object went massless at the discriminant point $P_\Delta$, not at $z=\infty$. The massless D-branes are not mutually local BPS states and so we would expect the physics to correspond to a four-dimensional conformal field theory. The candidate theory is described in \cite{Minahan:1996fg,Lerche:1996ni}.

An example for the global geometry was given in section \ref{ss:qeg}.
Note that the number of massless D-branes predicted by a Higgs-Coulomb
transition following (\ref{eq:Higgs}) would be 12. This is not correct
and the reason is that we have a nontrivial conformal field theory.
The actual changes in Hodge numbers associated with the collapse of a
del Pezzo surface $dP_n$ is given by the Coxeter number of $E_n$ as
observed in \cite{MV:F2,MS:five}. The relevant argument for this case
is explained in \cite{Lerche:1996ni}.

\subsubsection{$dP_7$}  \label{sss:dP7}

Now suppose $E=dP_7$, a del Pezzo surface of degree 2. This can be written as a double cover of $\P^2$ branched over a quartic curve, or equivalently a degree 4 surface in the weighted projective space $\P^3_{2,1,1,1}$.

The analysis of this is essentially identical to the case of $dP_6$, except now proposition \ref{thm:Kuz1} gives $S_{\A}^4=[6]$. Thus
\begin{equation*}
T_\infty^4 F = F S_\A^{-4} [8] = F \mathsf [2].
\end{equation*}
So all the D-branes in the image of $F$ are again $\Q$-massless and all the central charges go like $y^\frac14$.

As an example, let the matrix $Q$ be given by
\begin{equation*}
\begin{array}{c|ccccccc}&x_0&x_1&x_2&x_3&x_4&x_5&x_6\\\hline
D&-4&-1&2&1&1&0&1\\H&-2&1&0&0&0&1&0\end{array}.
\end{equation*}
The superpotential is $W=x_0f$, where (using previous conventions)
\begin{equation}
f = x_5^2(x_2^2+x_3^4+x_4^4+x_6^4) + x_1^2(x_2^3+x_3^6+x_4^6+x_6^6) \label{eq:dP7}
\end{equation}
The model has four phases. One of these, which we denote $\Sigma_+$, is a \CY\ phase with irrelevant ideal $B=(x_1,x_5)(x_2,x_3,x_4,x_6)$.
Another, which we denote $\Sigma_-$, is an exoflop phase with $B=(x_0,x_1)\cap(x_1,x_5)
\cap(x_2,x_3,x_4,x_5,x_6)$. The underlying topological space of $X_{\Sigma_-}$ has two irreducible components: a $\P^1$ with homogeneous coordinates $x_0,x_1$, and the singular hypersurface in
$\P^4_{2,1,1,1,1}$ given by $f=0$ with $x_1=1$. We thus have an extremal transition to the smooth degree 6 hypersurface in $\P^4_{2,1,1,1,1}$ from this exoflop phase.

Passing from $\Sigma_+$ to $\Sigma_-$, the surface $x_1=0$ (and thus we may set $x_5=1$) in $X_{\Sigma_+}$ contracts to a point. From \eqref{eq:dP7} we see that this surface is a $dP_7$.

\subsubsection{$dP_8$}  \label{sss:dP8}

Now suppose $E=dP_8$, a del Pezzo surface of degree 1. This can be written as a surface of degree 6 in the weighted projective space $\P^3_{3,2,1,1}$.

The analysis of this is essentially identical to the case of $dP_6$, except now we have $S_{\A}^6=[10]$, so
\begin{equation*}
T_\infty^6 F = F S_\A^{-6} [12] = F \mathsf [2].
\end{equation*}
So all the D-branes in the image of $F$ are again $\Q$-massless and all the central charges go like $y^\frac16$.

As an example, let the matrix $Q$ be given by
\begin{equation*}
\begin{array}{c|ccccccc}&x_0&x_1&x_2&x_3&x_4&x_5&x_6\\\hline
D&-6&3&0&1&1&-1&2\\H&0&1&-1&0&0&-1&1\end{array}.
\end{equation*}
The superpotential is $W=x_0f$, where (using previous conventions)
\begin{equation}
f = x_1^2x_2^2+x_1^3x_5^3+x_2^3x_6^3+x_3^6+x_4^6+x_5^6x_6^6 \label{eq:dP8}
\end{equation}
The model has six phases. One of these, which we denote $\Sigma_+$, is a \CY\ phase with $B=(x_2,x_5)(x_1,x_3,x_4,x_6)$.  Another,
which we denote $\Sigma_-$, is an exoflop phase with $B=(x_0,x_5)\cap(x_2,x_5)
\cap(x_1,x_2,x_3,x_4,x_6)$. The underlying topological space of $X_{\Sigma_-}$ has two irreducible components: a $\P^1$ with homogeneous coordinates $x_0,x_5$, and the singular hypersurface in
$\P^4_{2,1,1,1,1}$ given by $f=0$ with $x_5=1$. We thus have an extremal transition to the smooth degree 6 hypersurface in $\P^4_{2,1,1,1,1}$ from this exoflop phase.

Passing from $\Sigma_+$ to $\Sigma_-$, the surface $x_5=0$ (and thus we may set $x_2=1$) in $X_{\Sigma_+}$ contracts to a point. From \eqref{eq:dP8} we see that this surface is a $dP_8 $.

\subsubsection{$dP_5$}  \label{sss:dP5}

Something new happens when we let $E=dP_5$, a del Pezzo surface of degree 4. This cannot be written as a hypersurface in a weighted projective space, although it can be written as a complete intersection of two quadrics in $\P^4$.  The subcategory $\A = (\O_E)^\perp \subset \DC(E)$ is not a category of graded matrix factorizations, i.e.\ of D-branes in a \LG\ orbifold theory, but rather is a hybrid model over $\P^1$ with quadratic \LG\ fibre.

The subcategory $\A$ is not fractional \CY, as we now see by looking at the action of $S_\A$ on K-theory.  Let $\ell_1, \dotsc, \ell_5 \subset E$ be the exceptional lines of the blowup $E \to \P^2$ and $h$ the preimage of a line in $\P^2$ that misses the blown-up points, so $\O_E(h)$ the pullback of $\O_{\P^2}(1)$.  Then
\begin{equation} \label{eq:dP5_exc_coll}
\O_{\ell_1}(-1), \dotsc, \O_{\ell_5}(-1), \O_E(-2h), \O_E(-h)
\end{equation}
is a full exceptional collection for $\A$, hence gives a basis for $K(\A)$.  In this basis the Euler pairing is
\begin{equation*}
\chi = \left(\begin{smallmatrix}
1 & 0 & 0 & 0 & 0 & -1 & -1 \\
& 1 & 0 & 0 & 0 & -1 & -1 \\
& & 1 & 0 & 0 & -1 & -1 \\
& & & 1 & 0 & -1 & -1 \\
& & & & 1 & -1 & -1 \\
& & & & & 1 & 3 \\
& & & & & & 1 
\end{smallmatrix}\right).
\end{equation*}
To calculate the Serre functor in this basis, observe that for any $v, w \in K(\A)$ we have $\chi(v, S_\A(w)) = \chi(w, v)$, so in matrix terms $\chi \cdot S_\A = \chi^\top$, so
\begin{equation*}
S_\A = \chi^{-1} \cdot \chi^\top = \left(\begin{smallmatrix}
2 & 1 & 1 & 1 & 1 & -5 & -2 \\
1 & 2 & 1 & 1 & 1 & -5 & -2 \\
1 & 1 & 2 & 1 & 1 & -5 & -2 \\
1 & 1 & 1 & 2 & 1 & -5 & -2 \\
1 & 1 & 1 & 1 & 2 & -5 & -2 \\
2 & 2 & 2 & 2 & 2 & -8 & -3 \\
-1 & -1 & -1 & -1 & -1 & 3 & 1
\end{smallmatrix}\right)
\end{equation*}
The Jordan canonical form of this matrix is
\begin{equation}
\left(\begin{smallmatrix}
-1 & 1 \\
& -1 \\
& & 1 \\
& & & 1 \\
& & & & 1 \\
& & & & & 1 \\
& & & & & & 1
\end{smallmatrix}\right), \label{eq:jordan_form}
\end{equation}
so no power of $S_\A$ acts as $\pm 1$ on K-theory, so no power of $S_\A$ is a shift. Unlike the previous cases, not all of $\A$ can be \fluffy.

There is, however, a 6-dimensional subspace of $K(\A)$ on which $S_\A^2$ acts as the identity; we will exhibit ten objects of $\A$ whose K-theory classes span this subspace, and on which $S_\A^2$ acts a shift by 2, giving as many \fluffy\ objects on $X$ as possible.

The Jordan canonical form indicates that the remaining direction in $K(\A)$ should correspond to a central charge of the form
\[ \cz = \sqrt{y}\bigl(\log(y) + f(y)\bigr) \]
near $y=z^{-1}=0$, where $f(y)$ is a power series in $y$. Thus, this central charge also corresponds to a massless D-brane as $y\to0$. The monodromy of these ``log-massless'' D-branes can also be probed at the derived category level, as we see shortly. First we analyze the \fluffy\ objects.

\paragraph{\fluffy\ objects.}
\begin{prop}
The line bundles
\begin{equation} \label{eq:ten_bundles}
\begin{split}
\O_E(\ell_1 - h) &\qquad \O_E(\ell_2 + \ell_3 + \ell_4 + \ell_5 - 2h) \\
\O_E(\ell_2 - h) &\qquad \O_E(\ell_1 + \ell_3 + \ell_4 + \ell_5 - 2h) \\
\O_E(\ell_3 - h) &\qquad \O_E(\ell_1 + \ell_2 + \ell_4 + \ell_5 - 2h) \\
\O_E(\ell_4 - h) &\qquad \O_E(\ell_1 + \ell_2 + \ell_3 + \ell_5 - 2h) \\
\O_E(\ell_5 - h) &\qquad \O_E(\ell_1 + \ell_2 + \ell_3 + \ell_4 - 2h).
\end{split}
\end{equation}
lie in $\A$, that is, they have no cohomology.  Their classes in $K(\A)$ span the 1-eigenspace of $S_\A^2$.
\end{prop}
\begin{proof}
For the left-hand column, we have a short exact sequence
\begin{equation*}
0 \to \O_E(-h) \to \O_E(\ell_i - h) \to \O_{\ell_i}(\ell_i - h) \to 0.
\end{equation*}
Since $\O_{\ell_i}(\ell_i - h) = \O_{\ell_i}(-1)$, both ends of the short exact sequence are in the exceptional collection \eqref{eq:dP5_exc_coll}, so the middle term is in $\A$ as well.

For the right-hand column, recall that $\omega_E = \O_E(-3h + \ell_1 + \dotsb + \ell_5)$, so by Serre duality we have
\begin{equation*}
H^*(\O_E(\ell_2 + \ell_3 + \ell_4 + \ell_5 - 2h)) = H^{2-*}(\O_E(\ell_1 - h))^* = 0
\end{equation*}
and so on.

The claim about spanning the 1-eigenspace is a straightforward calculation.
\end{proof}

\begin{prop} \label{prop:with_the_square}
The Serre functor $S_\A$ acts on the line bundles \eqref{eq:ten_bundles} by exchanging the columns and shifting by 1: that is,
\begin{align*}
S_\A(\O_E(\ell_1 - h)) &= \O_E(\ell_2 + \ell_3 + \ell_4 + \ell_5 - 2h)[1] \\
S_\A(\O_E(\ell_2 + \ell_3 + \ell_4 + \ell_5 - 2h)) &= \O_E(\ell_1 - h)[1]
\end{align*}
and so on.  In particular $S_\A^2$ acts on each of the line bundles \eqref{eq:ten_bundles} as a shift by 2.
\end{prop}
\begin{proof}
Let $p_1, \dotsc, p_5 \in \P^2$ be the five points blown up by the map $E \to \P^2$.  Let $\gamma_1 \subset E$ be the proper transform of the line joining $p_1$ and $p_2$, and $\gamma_2 \subset E$ the proper transform of the unique conic in $\P^2$ passing through all five points.  As divisor classes on $E$ we have
\begin{align*}
h &= \gamma_1 + \ell_1 + \ell_2 \\
2h &= \gamma_2 + \ell_1 + \dotsb + \ell_5.
\end{align*}
The lines $\ell_1$, $\ell_2$, $\gamma_1$, and $\gamma_2$ intersect in a square as shown:
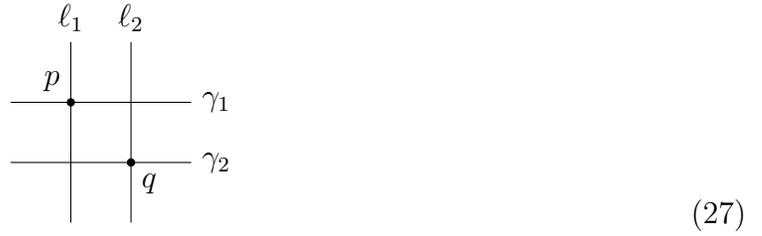
\begin{equation}
\begin{tikzpicture}[scale=0.8]
  \draw (0,1) -- (3,1) node[anchor=west]{$\gamma_1$};
  \draw (0,0) -- (3,0) node[anchor=west]{$\gamma_2$};
  \draw (1,-1) -- (1,2) node[anchor=south]{$\ell_1$};
  \draw (2,-1) -- (2,2) node[anchor=south]{$\ell_2$};
  \filldraw (2,0) circle (0.06) node[anchor=north west]{$q$};
  \filldraw (1,1) circle (0.06) node[anchor=south east]{$p$};
\end{tikzpicture} \label{eq:LM}
\end{equation}

Now the two line bundles in question are
\begin{gather*}
\O_E(\ell_1 - h) = \O_E(-\ell_2 - \gamma_1) \\
\O_E(\ell_2 + \ell_3 + \ell_4 + \ell_5 - 2h) = \O_E(-\ell_1 - \gamma_2).
\end{gather*}
First we claim that the dual line bundles are globally generated and have $h^0 = 2$.  For definiteness we argue with $\O_E(\ell_1 + \gamma_2)$.  Since $\gamma_2$ has self-intersection $-1$ we can contract it; then $\ell_1$ becomes a line of self-intersection 0 and hence the fiber of a map to $\P^1$.  Thus $\O_E(\ell_1 + \gamma_2)$ is the pullback of $\O_{\P^1}(1)$, hence is globally generated by two sections.

Next we claim that $S_\A^{-1}(\O_E(-\ell_2 - \gamma_1)) = \O_E(-\ell_1 - \gamma_2)[-1]$.  We have $S_\A^{-1} = L_{\O_E} \circ S_E^{-1}$, where $L_{\O_E}$ is left mutation past $\O_E$ as in \eqref{eq:mutation}.  Since
\begin{align*}
\omega_E &= \O_E(-3h + \ell_1 + \dotsb + \ell_5) \\
&= \O_E(-\ell_1 - \ell_2 - \gamma_1 - \gamma_2),
\end{align*}
we see that
\begin{equation*}
S_E^{-1}(\O_E(-\ell_2 - \gamma_1)) = \O_E(\ell_1 + \gamma_2)[-2].
\end{equation*}
Since $\O_E(\ell_1 + \gamma_2)$ is globally generated by two sections, the evaluation map
\begin{equation*}
\RHom(\O_E, \O_E(\ell_1 + \gamma_2)) \otimes \O_E \to \O_E(\ell_1 + \gamma_2)
\end{equation*}
is a surjection $\O_E^2 \to \O_E(\ell_1 + \gamma_2)$, so the kernel is a line bundle; taking first Chern classes we see that it must be $\O_E(-\ell_1 - \gamma_2)$.  Thus
\begin{equation*}
L_{\O_E} \O_E(\ell_1 + \gamma_2)[-2] = \O_E(-\ell_1 - \gamma_2)[-1],
\end{equation*}
as claimed.

Similarly we find that $S_\A^{-1}(\O_E(-\ell_1 - \gamma_2)) = \O_E(-\ell_2 - \gamma_1)[-1]$.
\end{proof}

\paragraph{Log-massless objects.}
Having accounted for the \fluffy\ objects, we now turn to the remaining direction in $K(\A)$.  From the the non-trivial Jordan block of \eqref{eq:jordan_form} we see that it is spanned by a class $v$ with
\begin{equation}
S_\A(v) = -v + w \label{eq:jordan-blocked}
\end{equation}
for some $(-1)$-eigenvector $w$ of $S_\A$.  Indeed $v = [\cI_\text{point}]$ is such a class, and we can upgrade \eqref{eq:jordan-blocked} from K-theory to the derived category as follows:

\begin{prop}
Let $\gamma_1$ and $\gamma_2$ be the exceptional lines appearing in the proof of proposition \ref{prop:with_the_square}, and consider the points $p = \ell_1 \cap \gamma_1$ and $q = \ell_2 \cap \gamma_2$.  Then there is an exact triangle
\[ S_\A(\cI_p) \to \O_E(-\ell_1 - \gamma_2)[2] \oplus \O_E(-\ell_2 - \gamma_1)[2] \to \cI_q[2] \]
and similarly
\[ S_\A(\cI_q) \to \O_E(-\ell_1 - \gamma_2)[2] \oplus \O_E(-\ell_2 - \gamma_1)[2] \to \cI_p[2]. \]
Note that $S_\A$ acts on the middle terms of these triangles as a shift by 1.
\end{prop}
\begin{proof}
Apply $L_{\O_E}(- \otimes \omega_E^{-1})$ to the exact sequence
\[ 0 \to \cI_{p \cup q} \to \cI_q \to \O_p \to 0 \]
to get an exact triangle
\begin{equation} \label{eq:intermediate_triangle}
L_{\O_E}(\cI_{p \cup q} \otimes \omega_E^{-1}) \to S_\A^{-1} \cI_q[2] \to \cI_p[1].
\end{equation}
For the left-hand term, take the Mayer--Vietoris sequence
\[ 0 \to \cI_{\ell_1 \cup \ell_2 \cup \gamma_1 \cup \gamma_2}
\to \cI_{\ell_1 \cup \gamma_2} \oplus \cI_{\ell_2 \cup \gamma_1}
\to \cI_{(\ell_1 \cup \gamma_2) \cap (\ell_2 \cup \gamma_1)} \to 0 \]
and write it as
\[ 0 \to \omega_E \to \O_E(-\ell_1 - \gamma_2) \oplus \O_E(-\ell_2 - \gamma_1) \to \cI_{p \cup q} \to 0; \]
then $L_{\O_E}(- \otimes \omega_E^{-1})$ annihilates the first term, and as we saw in the proof of proposition \ref{prop:with_the_square} it acts on the second term as a shift by 1, so we have
\[ L_{\O_E}(\cI_{p \cup q} \otimes \omega_E^{-1}) = \O_E(-\ell_1 - \gamma_2)[1] \oplus \O_E(-\ell_2 - \gamma_1)[1]. \]
Substitute this into \eqref{eq:intermediate_triangle} and rotate the triangle to get
\[ \cI_p \to \O_E(-\ell_1 - \gamma_2)[1] \oplus \O_E(-\ell_2 - \gamma_1)[1] \to S_\A^{-1} \cI_q[2]. \]
Now all three objects are in $\A$; apply $S_\A$ to get the desired result.
\end{proof}

\paragraph{Example.}
Finally we give an example.  Let the matrix $Q$ be given by
\begin{equation*}
\begin{array}{c|ccccccccc}&x_0&x_1&x_2&x_3&x_4&x_5&x_6&x_7&x_8\\\hline
D&-2&-2&-1&0&1&1&1&1&1\\H&-1&-1&1&1&0&0&0&0&0\end{array}.
\end{equation*}
The superpotential is $W=x_0f_0+x_1f_1$, where both $f_0$ and $f_1$ are of the form (using previous conventions)
\begin{equation*}
x_2(x_4^3+x_5^3+x_6^3+x_7^3+x_8^3) + x_3(x_4^2+x_5^2+x_6^2+x_7^2+x_8^2) 
\end{equation*}
The model has four phases. One of these, which we denote $\Sigma_+$, is a \CY\ phase with irrelevant ideal $B=(x_2,x_3)(x_4,x_5,x_6,x_7,x_8)$.
The Hodge numbers are $h^{1,1}=2$ and $h^{2,1}=66$.  Another phase, which we denote $\Sigma_-$, is an exoflop with $B=(x_2,x_3)\cap(x_0,x_1,x_2)\cap(x_3,x_4,x_5,x_6,x_7,x_8)$. To pass from $\Sigma_+$ to $\Sigma_-$ we shrink down the surface $x_2=0$ which is $dP_5$. The underlying topological space of $X_{\Sigma_-}$ has two irreducible components: a $\P^2_{\{2,2,1\}}$ with homogeneous coordinates $x_0,x_1,x_2$, and a singular complete intersection of two cubics in
$\P^5$ given by $f_0=f_1=0$ with $x_2=1$. We thus have an extremal transition to the bicubic in $\P^5$ from this exoflop phase.

\subsubsection{$dP_4$ and lower}  \label{sss:dP4}

Repeating the K-theory analysis of the last section for all $dP_n$, we find that the characteristic polynomial of $S_\A$ is
\begin{align*}
dP_1: &\qquad (\lambda-1)(\lambda^2 + 6\lambda + 1) \\
dP_2: &\qquad (\lambda-1)^2(\lambda^2 + 5\lambda + 1) \\
dP_3: &\qquad (\lambda-1)^3(\lambda^2 + 4\lambda + 1) \\
dP_4: &\qquad (\lambda-1)^4(\lambda^2 + 3\lambda + 1) \\
dP_5: &\qquad (\lambda-1)^5(\lambda^2 + 2\lambda + 1) \\
dP_6: &\qquad (\lambda-1)^6(\lambda^2 + \lambda + 1) \\
dP_7: &\qquad (\lambda-1)^7(\lambda^2 + 1) \\
dP_8: &\qquad (\lambda-1)^8(\lambda^2 - \lambda + 1)
\end{align*}
Thus for $n \le 4$ we see that $S_\A$ has two negative real eigenvalues, one with $\lambda < -1$ and one with $-1 < \lambda < 0$, so $\A$ is even further from being fractional \CY.  Despite the large 1-eigenspace in K-theory, there does not seem to be any object of $\A$ on which $S_\A$ acts a shift.

For $dP_6$, the eigenvalues are 1 and the primitive cube roots of unity, reflecting the fact that $S_\A^3 = [4]$.  Moreover, despite having a 6-dimensional 1-eigenspace in K-theory, there is no object of $\A$ on which $S_\A$ acts a shift, for this would contradict $S_\A^3 = [4]$.

For $dP_7$ the eigenvalues are 1 and $\pm i$, reflecting the fact that $S_\A^4 = [6]$.  But note that $S_\A^2 \ne [3]$, since in K-theory $S_\A^2$ has both 1 and $-1$ as eigenvalues; that is, $\A$ is fractional \CY\ of dimension $6/4$ but not $3/2$.

For $dP_8$ the eigenvalues are 1 and the primitive sixth roots of unity, reflecting the fact that $S_\A^6 = [10]$, and again $S_\A^3 \ne [5]$.

\subsection{Other Contractions}

It is interesting to view other well-known extremal transitions in terms of exoflops and spherical functors. We will see that \fluffy ness coincides with known results for massless D-branes. 

\subsubsection{$dP_{10}$} \label{sss:dP10}

As an example, let the matrix $Q$ be given by
\begin{equation*}
\begin{array}{c|ccccccc}&x_0&x_1&x_2&x_3&x_4&x_5&x_6\\\hline
D&-2&-1&0&0&1&1&1\\H&-3&1&1&1&0&0&0\end{array}.
\end{equation*}
The superpotential is $W=x_0f$, where (using previous conventions)
\begin{equation*}
x_1^3(x_4^5+x_5^5+x_6^5) + (x_2^3+x_3^3)(x_4^2+x_5^2+x_6^2)
\end{equation*}
The model has four phases. The only \CY\ phase, which we denote $\Sigma_+$, has irrelevant ideal $B=(x_1,x_2,x_3)(x_4,x_5,x_6)$.
The Hodge numbers are $h^{1,1}=2$ and $h^{2,1}=86$.
An exoflop phase,  which we denote $\Sigma_-$, has $B=(x_0,x_1)\cap(x_1,x_2,x_3)\cap(x_2,x_3,x_4,x_5,x_6)$. 
To pass from $\Sigma_+$ to $\Sigma_-$ we shrink down the surface $x_1=0$, which we call $E$.

This $E$ is a surface of bidegree $(3,2)$ in $\P^1_{[x_2,x_3]} \times \P^2_{[x_4,x_5,x_6]}$, hence is a conic bundle over $\P^1$ with nine singular fibres as shown in figure \ref{fig:dP10}: the general fibre is a $\P^1$ of self-intersection 0, but over 9 points in the base $\P^1$ the fibre splits into two $\P^1$'s of self-intersection $-1$, which we call $\ell_i$ and $\gamma_i$.  In fact $E$ is $dP_{10}$.\footnote{Sections of this conic fibration may be found by setting $x_4,x_5,x_6$ as quadric functions of $x_2$ and $x_3$. This gives a section missing all the $\ell_i$ and having self-intersection $-1$.}  Passing from $\Sigma_+$ to $\Sigma_-$ collapses $E$ to the base $\P^1$.

\begin{figure}
\begin{center}
\begin{tikzpicture}[scale=0.7]
\foreach \x in {-6,-5,-3,-1,1,3,4}
  \draw (\x,-1) .. controls (\x-0.3,-0.5) and (\x+0.3,0.5) .. (\x,1);
\foreach \x / \y in {-4/1,-2/2,2/9}
{
  \draw (\x,-1) .. controls (\x,-0.3) .. (\x+0.2,0.2);
  \draw (\x,-1) node[anchor=north]{$\scriptstyle\gamma_{\y}$};
  \draw (\x,1) node[anchor=south]{$\scriptstyle\ell_{\y}$};
  \draw (\x,1) .. controls (\x,0.3) .. (\x+0.2,-0.2);
}
\draw (0,0) node {$\cdots$};
\draw (-1,-2) node {$\downarrow$};
\draw (-6,-3) -- (4,-3);
\end{tikzpicture}
\end{center}
\caption{$dP_{10}$} \label{fig:dP10}
\end{figure}

A physics analysis along the lines of \cite{AKM:lcy} implies we have 9 hypermultiplets in the fundamental representation of $\SU(2)$ given by the 9 degenerate conic fibers. We would like to do better than this and give the precise objects in $\DC(X)$ corresponding to these D-branes. We will show that the following 18 objects are indeed \fluffy:
\begin{equation} \label{eq:18_massless_branes}
\begin{split}
\O_{\ell_1}(-1), &\dotsc, \O_{\ell_9}(-1), \\
\O_{\gamma_1}(-1), &\dotsc, \O_{\gamma_9}(-1).
\end{split}
\end{equation}

As usual we have $E = H-D \in \Pic(X)$ and $T_0 = \O_X(D) \otimes -$.  But now $T_\Delta$ is the Horja twist associated to $i_* q^*: \DC(\P^1) \to \DC(X)$, where $q: E \to \P^1$ is the conic bundle.  Write
\begin{align*}
\DC(E) &= \langle \A, q^* \DC(\P^1) \rangle \\
&= \langle \A, \O_E, \O_E(f) \rangle
\end{align*}
where $f$ is the divisor class of a conic fiber.  Then
\begin{align*}
T_\infty &= T_\Delta T_0 \\
&= T_{q^* \DC(\P^1)}(\O_X(D) \otimes -) \\
&= T_\A^{-1} T_{\DC(E)} (\O_X(D) \otimes -) \\
&= T_\A^{-1} (\O_X(E+D) \otimes -) \\
&= T_\A^{-1} (\O_X(H) \otimes -).
\end{align*}
But beware that $H|_E = f$, whereas in previous examples we had $H|_E = 0$.

\begin{prop}
$T_\infty$ acts on the sheaves \eqref{eq:18_massless_branes} by exchanging the rows and shifting by 1: that is,
\begin{align}
T_\infty \O_{\ell_i}(-1) &= \O_{\gamma_i}(-1)[1] \label{eq:going_to_prove} \\
T_\infty \O_{\gamma_i}(-1) &= \O_{\ell_i}(-1)[1] \label{eq:not_going_to_prove}
\end{align}
In particular $T_\infty^2$ acts on each of the sheaves \eqref{eq:18_massless_branes} as a shift by 2.
\end{prop}
\begin{proof}
We prove \eqref{eq:going_to_prove}; the proof of \eqref{eq:not_going_to_prove} is entirely similar. We include the usually implicit $i_*$ for clarity. First of all we have
\[ \O_X(H) \otimes i_* \O_{\ell_i}(-1) = i_*(\O_E(f) \otimes \O_{\ell_i}(-1)) = i_* \O_{\ell_i}(-1). \]
Since $\O_{\ell_i}(-1) \in \A$, we have
\begin{align*}
T_\A^{-1} i_* \O_{\ell_i}(-1)
&= i_* S_\A^{-1} \O_{\ell_i}(-1)[2] \\
&= i_* L_{\O_E} L_{\O_E(f)} S_E^{-1} \O_{\ell_i}(-1)[2] \\
&= i_* L_{\O_E} L_{\O_E(f)}(\O_{\ell_i}(-1) \otimes \omega_E^{-1}) \\
&= i_* L_{\O_E} L_{\O_E(f)} \O_{\ell_i}
\end{align*}
where in the last step we used the adjunction formula.  To compute the mutations, first we have
\[ \RHom_E(\O_E(f), \O_{\ell_i}) = \RHom_E(\O_E, \O_{\ell_i}) = \C, \]
so
\[ L_{\O_E(f)} \O_{\ell_i} = \Cone(\O_E(f) \to \O_{\ell_i}) = \O_E(f-\ell_i)[1] = \O_E(\gamma_i)[1]. \]
Next, from the exact sequence
\[ 0 \to \O_E \to \O_E(\gamma_i) \to \O_{\gamma_i}(-1) \to 0 \]
we see that
\[ \RHom_E(\O_E, \O_E(\gamma_i)) = \C \]
and so
\[ L_{\O_E} \O_E(\gamma_i)[1] = \Cone(\O_E[1] \to \O_E(\gamma_i)[1]) = \O_{\gamma_i}(-1)[1] \]
as claimed.
\end{proof}

The underlying topological space of $X_{\Sigma_-}$ has two irreducible components: a $\P^1\times\P^1$ with homogeneous coordinates $x_0,x_1$ and $x_2,x_3$, and a singular quintic in
$\P^4$. The exoflop phase gives an extremal transition to the quintic threefold. 

Note that the 18 objects (\ref{eq:18_massless_branes}) do not generate a category of 18 distinct points. In particular we have $\Ext^1(\O_{\ell_i},\O_{\gamma_i})=\C$. It is natural to ask, therefore, whether this extension also corresponds to a massless D-brane. Actually it does not because the K-theory classes of $\O_{\ell_i}$ and $\O_{\gamma_i}$ are identical in $X$ and thus their ``slope'' $\xi$'s are identical. By the rules of {\em polystability\/}, the only stable object given by this extension would be the trivial $\O_{\ell_i}\oplus\O_{\gamma_i}$. 

Connecting with physics, we note the change in Hodge numbers from 18 massless D-branes is again consistent with the Higgs mechanism since the $\SU(2)$ breaking needs to ``eat'' 3 massless D-branes. Also note that the $\SU(2)$ gauge symmetry actually appears at the discriminant point $P_\Delta$ in figure \ref{fig:mon1} rather than $z=\infty$. In fact, a local analysis zooming in on $P_\Delta$ is identical to that of \cite{KV:N=2} and we get an $\SU(2)$ theory with {\em no\/} flavours. The 9 flavours we have have nonzero, but identical, ``bare'' mass parameters, $m_i$, and all become massless at $z=\infty$ away from the enhanced $\SU(2)$ symmetry.

\subsubsection{$dP_{16}$} \label{sss:dP16}

Last we consider the case that $E$ is $\P^2$ blown up at $n$ points (possibly more than 8, taking it out of the class of del Pezzo surfaces) and the perestroika consists of the blowing down of this surface back to $\P^2$.

So 
\begin{equation*}
\DC(E) = \langle\O_{\ell_1}(-1),\O_{\ell_2}(-1),\ldots,\O_{\ell_n}(-1),\DC(\P^2)\rangle.
\end{equation*}
According to the EZ-transformation picture, $i_*q^*\DC(Z)=\DC(\P^2)$. 
This is the monodromy in the wall separating the phases. So the monodromy around the limit at $z=\infty$ is given by the remainder:
\begin{equation*}
  \A = \langle \O_{\ell_1}(-1),\O_{\ell_2}(-1),\ldots,\O_{\ell_n}(-1)\rangle.\end{equation*}
These objects are all orthogonal to one another so we have the derived category of $n$ distinct points. The situation is similar to section \ref{sss:P1P1} and we have $n$ \fluffy\ objects in this limit.

As an example, let the matrix $Q$ be given by
\begin{equation*}
\begin{array}{c|ccccccc}&x_0&x_1&x_2&x_3&x_4&x_5&x_6\\\hline
D&-1&0&0&0&1&1&-1\\H&-5&1&1&1&1&1&0\end{array}.
\end{equation*}
The superpotential is $W=x_0f$, where $f$ is of the form (using previous conventions)
\begin{equation*}
(x_4+x_5)(x_1^4+x_2^4+x_3^4) + x_6^4(x_4^5+x_5^5)
\end{equation*}
The model has four phases. The only \CY\ phase, which we denote $\Sigma_+$, has irrelevant ideal $B=(x_4,x_5)(x_1,x_2,x_3,x_6)$.
The Hodge numbers are again $h^{1,1}=2$ and $h^{2,1}=86$.  Another phase, which we denote $\Sigma_-$, is an exoflop with $B=(x_0,x_6)\cap(x_1,x_2,x_3,x_6)\cap(x_1,x_2,x_3,x_4,x_5)$. 

To pass from $\Sigma_+$ to $\Sigma_-$ we shrink down the surface $x_6=0$ which we denote $E$.  This $E$ maps to a $\P^2$ with homogeneous coordinates $x_1,x_2,x_3$, where the fibre is generically a point. At 16 points, however, the fibre is a full $\P^1$. Thus $E$ is $dP_{16}$. 

The underlying topological space of $X_{\Sigma_-}$ has 17 irreducible components: 16 consist of $\P^1$ with homogeneous coordinates $x_0,x_6$, and the final component is a singular quintic in $\P^4$. This quintic has 16 distinct singular points, and at each one we have a $\P^1$ sticking out:
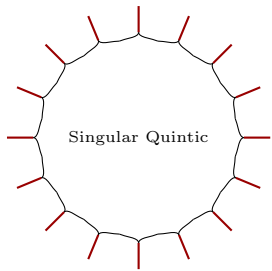
\begin{center}
\begin{tikzpicture}[scale=0.7]
\foreach \x in {0,22.5,45,67.5,90,112.5,135,157.5,180,202.5,225,247.5,
                270,292.5,315,337.5}
{
  \draw[thick,color=red!60!black] (\x:2) -- (\x:2.5);  
  \draw (\x:2) .. controls (\x:1.9) and (\x+8:1.9) .. (\x+11.25:1.87);
  \draw (\x:2) .. controls (\x:1.9) and (\x-8:1.9) .. (\x-11.25:1.87);
}
\draw (0,0) node{\tiny Singular Quintic};
\end{tikzpicture}
\end{center}
 We again have an extremal transition to the quintic threefold. This is the old favorite ``conifold'' transition studied originally in \cite{GH:con} and also viewed in same context as we are seeing it in \cite{GMV:Ht}. The \fluffy\ objects we see are exactly consistent with this picture.

It is amusing to note that the exoflop picture of a conifold transition has the ``conifold'' in the exoflop limit, while the wall between the phases is associated with the spherical functor from $\DC(\P^2)$. The latter is a ``large-radius limit'' as discussed in example \ref{eg:divisor} in section \ref{s:sphF}. So a ``large-radius limit'' actually lives in the wall between the phases! This latter large radius \CY\ is, of course, a flop of the \CY\ phase $\Sigma_+$.


\section{Discussion} \label{s:disc}

Clearly it would be nice to elucidate the precise relationship between \fluffy ness and actual masslessness. In the Seidel--Thomas twist on $\O_E$ we expect, from physics, to have precisely one massless object. The object $\O_E$ itself is \fluffy\ but it would be nice to prove this is the only \fluffy\ object.
This could then be extended to the cases of $\P^1\times\P^1$ and $dP_{16}$ in the text where we had a finite number of massless objects.

In the cases where we have an infinite number of \fluffy\ objects we are probably over-counting the massless states. For example, if $\a$ is an object in $\A$ and $F:\A\to X$ is the inclusion, then $\Ext^3_X(F\a,
F\a)$ is nonzero. The mapping cone of this morphism is then \fluffy\ but unstable. It is therefore necessary to restrict attention to objects purely in the image of $F$. This might suggest the following
\begin{question}
If $\a \in \A$ is a simple object which is \fluffy\ at $P\in\cMcp$, is $F\a$ always then a stable D-brane with $\lim \cz=0$ at this point?
\end{question}

It would also be interesting to further understand the connection between stability and \fluffy ness. As discussed in section \ref{ss:stab}, the example of $\P^2$ shows that an object may be shown to be unstable by a monodromy argument. However, we have no general statements.

Another obvious direction to pursue is an extension of our analysis to the cases $dP_1$ to $dP_4$.  For $dP_1$ to $dP_3$, the del Pezzo surfaces are themselves toric but we would have to embed in \CY\ threefolds of Picard rank $\geq3$, which are a little harder to work with than the Picard rank 2 ones studied here.  For $dP_4$ the \CY\ threefold might have to be a complete intersection in a Grassmannian, for which the mirror is not well-understood.

More broadly, our analysis of $\DC(X)$ in an exoflop limit may allow one to ``follow'' the derived category through an extremal transition in the context of toric geometry. Since a huge number of \CY\ threefolds are connected via such transitions, this may provide more insight into the structure of the derived category for \CY\ threefolds.

\section*{Acknowledgments}

We thank R.~Plesser for many helpful discussions. PSA is supported by NSF grants DMS--0905923 and DMS--1207708. Any opinions, findings, and
conclusions or recommendations expressed in this material are those of
the authors and do not necessarily reflect the views of the National
Science Foundation.


\end{document}